\documentclass[11pt]{article}

\usepackage[numbers]{natbib}
\usepackage{amsmath}
\usepackage{amssymb}
\usepackage{amsthm}
\usepackage{xcolor}
\usepackage{graphicx}
\usepackage{array}
\usepackage{booktabs}
\usepackage{tabularx}
\usepackage{hyperref}
\usepackage{algorithm}
\usepackage{algpseudocode}
\usepackage{enumitem}\usepackage[margin=1in]{geometry}

\hypersetup{colorlinks=true,citecolor=blue,linkcolor=blue,hypertexnames=false}
\allowdisplaybreaks

\newcommand{\paperTitle}{An $n^{8/5+o(1)}$-Time $\Omega(\lambda^3)$-Approximation for Longest Common Subsequence}
\newcommand{\paperAuthor}{Zhao Song\thanks{\texttt{magic.linuxkde@gmail.com}. The author started working on this problem at the end of July 2026 and finished a draft of the $n^{1.6}$ result on August 25, 2026. The author had planned to release the draft only once the result was tight and could not be improved further. However, due to recent progress on FMM, the author decided to release this draft now. The author would like to thank Omri Weinstein for useful discussions.}}

\newcommand{\E}{\mathbb E}

\DeclareMathOperator{\lcs}{lcs}
\DeclareMathOperator{\lis}{lis}
\DeclareMathOperator{\opt}{opt}

\theoremstyle{plain}
\newtheorem{theorem}{Theorem}[section]
\newtheorem{lemma}[theorem]{Lemma}
\newtheorem{definition}[theorem]{Definition}

\newtheorem{corollary}[theorem]{Corollary}
\newtheorem{observation}[theorem]{Observation}
\theoremstyle{definition}

\newtheorem{remark}[theorem]{Remark}

\begin{document}

\date{August 25, 2026}
\title{\paperTitle}
\author{\paperAuthor}
\maketitle

\begin{abstract}
Let $\lambda$ denote the ratio of the length of a longest common subsequence of two
length-$n$ strings to $n$. Rubinstein, Seddighin, Song and Sun \cite{rsss19} gave an
$\Omega(\lambda^3)$-approximation for LCS running in $\widetilde O(n^{39/20})$ time,
where $39/20=1.95$. Song \cite{son19} mentioned that improving the $n^{1.95}$ running time is an interesting open question.

We give an algorithm that computes an $\Omega(\lambda^3)$-approximation
of the longest common subsequence in
$n^{8/5+o(1)}$ time.  This improves the
exponent $1.95$ to $1.6+o(1)$.

\end{abstract}
\newpage

\section{Introduction}

The longest common subsequence (LCS) of two strings of length $n$ is computed by a
textbook dynamic program in $O(n^2)$ time, and half a century of effort has
improved this only by logarithmic factors. Fine-grained complexity explains why:
no truly subquadratic algorithm exists unless the Strong Exponential Time
Hypothesis fails \cite{abw15,bk15}. Attention has therefore turned to
approximation, where the picture remains strikingly incomplete.

Two lines of work have emerged. The first asks for the best factor achievable in a
given running time, with no assumption on the input. In near-linear time the best
known factor is $\widetilde O(n^{0.4})$, and more generally an
$\widetilde O(n^{2\varepsilon/5})$-approximation is achievable in $O(n^{2-\varepsilon})$
time for every $0 < \varepsilon \le 1$ \cite{bcd21}.
Combining a sublinear-time estimator for longest increasing subsequences
\cite{anss22} with the classical reduction of \cite{hs77} has been proposed as a
route to an $n^{o(1)}$-approximation in linear time \cite{nos21}. At the other end of the range,
a $(1-\epsilon)$-approximation is known only in time $n^2/2^{\log^{\Omega(1)}n}$
\cite{mr26}, which improves on the quadratic dynamic program by a quasi-polynomial
factor but is not truly subquadratic.

Over fixed alphabets the situation is different. Rubinstein and Song gave a
$(1/2+\Omega(1))$-approximation for equal-length binary strings in truly
subquadratic time \cite{rs20}; Akmal and Vassilevska Williams extended this to
every constant alphabet \cite{avw21}; and He and Li obtained a
$(1/2+\Omega(1))$-approximation for binary strings in $n^{1+\varepsilon}$ time,
allowing unequal lengths \cite{hl23}. These results rely on a fixed alphabet and
do not apply to the unrestricted-alphabet setting considered here. For
unrestricted alphabets, Boneh, Golan and Kraus recently gave a deterministic
near-linear-time $O(n^{3/4}\log n)$-approximation \cite{bgk25}.

The second line parameterizes by the size of the optimum. Writing
$\lambda = \lcs(A,B)/n$, Rubinstein, Seddighin, Song and Sun \cite{rsss19} gave an
$\Omega(\lambda^3)$-approximation running in truly subquadratic time
$\widetilde O(n^{39/20})$, where $39/20 = 1.95$. \cite{son19} mentioned that improving the $n^{1.95}$ running time is an interesting open question. To our knowledge, no later
published work improves this $\lambda$-parameterized running-time/approximation
pair for unrestricted alphabets: at the same running time \cite{bcd21} returns
only an $n^{\Omega(1)}$-factor approximation, and \cite{mr26} is slower.

\subsection{Our result}
\label{sec:our-result}
We state our main result as follows.
\begin{theorem}[Main result; informal version of
Theorem~\ref{thm:uniform_eight_fifths_formal}]
\label{thm:uniform_eight_fifths_informal}
\label{thm:fixed_delta_informal}
There is a randomized algorithm that, given strings $A$ and $B$ of length
$n$ with $\lambda=\lcs(A,B)/n$, computes an $\Omega(\lambda^3)$-approximation
of $\lcs(A,B)$ in time $O(n^{8/5+o(1)})$ with high probability.
\end{theorem}
Our theorem improves the  best previous work \cite{rsss19} from $n^{1.95}$ to $n^{1.6+o(1)}$.

{\bf Organization of the paper.}
Section~\ref{sec:tech} reviews the algorithm of \cite{rsss19} and outlines the main ideas behind our improvement.
Section~\ref{sec:preliminaries} introduces the notation and collects the tools from previous work used in our analysis.
Section~\ref{sec:proof} develops the improved sparsification and multiscale repair procedures, analyzes their running times, and combines them with the small-density algorithms to prove the main theorem.

\section{Technique Overview}
\label{sec:tech}

Section~\ref{sec:tech-rsss} reviews the window framework, sparsification, and nearby-search steps of \cite{rsss19}, and explains how their costs yield the $n^{1.95}$ running time.
Section~\ref{sec:tech-ours} outlines our improvements to these steps and explains how adaptive sparsification, multiscale exact repair, and the small-density algorithms combine to give the $n^{8/5+o(1)}$ running time.

\subsection{Summary of RSSS19}
\label{sec:tech-rsss}

\paragraph{Windows and compatible paths.}
Fix a base length $d$. The set $W_A$ partitions $A$ into $n/d$ windows of length $d$. For each dyadic scale $d_i=d\cdot2^i$, the ladder $W_B$ contains substrings of $B$ starting at multiples of $d_i$ and having lengths in $(d_i/2,d_i]$ that are multiples of $d$, up to $w_{\max}=\Theta(d/\lambda)$. Hence there are $O(\log(1/\lambda))$ dyadic scales but $\Theta(1/\lambda)$ possible exact lengths. Theorem~3.21 of \cite{rsss19} handles unequal lengths through padded size-pair instances. In our implementation we avoid enumerating exact-length sublayers and instead retain the pair-specific value $D_{ij}:=\max\{|w_i|,|w_j|\}$ within each dyadic scale; see Section~\ref{sec:tech-ours}. Lemma~4.7 of \cite{rsss19} gives a compatible common subsequence of length $\Omega(\lambda n)-2\epsilon_0n$, where $\epsilon_0=\epsilon\lambda$, for one of $(A,B)$, $(B,A)$, $(\operatorname{rev}A,\operatorname{rev}B)$, and $(\operatorname{rev}B,\operatorname{rev}A)$, so the algorithm runs all four orientations and returns the largest certified value. Given lower bounds $M(w,w')\leq\lcs(w,w')$, a block dynamic program finds the best compatible value in $O(|W_A||W_B|)=\widetilde O(n^2/d^2)$ time.

\paragraph{Step 1: sparsification.}
For one equal-length or padded size-pair instance, put $\widetilde\lambda=\lambda^2/2$ and sample $\widetilde O(k^\gamma)$ centres $w_a$. For each centre and window $w_i$, compute a fixed LCS position set $\opt_{i,a}\subseteq w_a$ and a greedy packing $Y_{a,i}$ until the residual LCS falls below $|w_a|\widetilde\lambda/2$. The test $\|\lcs_{w_a}(w_i,w_j)\|\geq\widetilde\lambda$, where $\lcs_{w_a}(w_i,w_j)=\lcs(\opt_{i,a},w_j)$, certifies the lower bound $\lcs(w_i,w_j)/\max\{|w_i|,|w_j|\}\geq\lambda^4/16$ used in Theorem~3.21 of \cite{rsss19}. If at least $2/\lambda$ windows have normalized LCS at least $\lambda$ with one centre, Cauchy--Schwarz gives two position sets with overlap at least $\lambda^2/2$. Call a pair close if at least $k^{1-\gamma}$ centres witness it; $\widetilde O(k^\gamma)$ sampled centres find every close pair with high probability. If more than $2k^{2-\gamma}/\lambda$ pairs were missed, a maximum-degree vertex would have more than $2k^{1-\gamma}/\lambda$ missed neighbours. That neighbourhood has no independent set of size $2/\lambda$, so Tur\'an's theorem \cite{turan41} produces a vertex with at least $k^{1-\gamma}$ common neighbours, making the corresponding missed pair close, a contradiction. Thus at most $2k^{2-\gamma}/\lambda$ pairs are missed. Apart from the exact-size multiplicity above, Step~1 costs $\widetilde O(k^{1+\gamma}w_{\max}^2+k^{2+\gamma}w_{\max})$.

\paragraph{Step 2: nearby search.}
Write $k^{2-\eta}$ for the missed-pair bound. Following \cite{cdgks18}, windows on the same side are nearby when their starts differ by at most $W:=k^{\eta/2}w_{\max}$. Sample $A$-rows with probability $p=\min\{1,\widetilde O(k^{-\eta/2}/\epsilon_{\mathrm{nbs}})\}$, inspect every sampled row exactly, and recompute all nearby pairs around each discovered miss. Along a fixed optimal compatible path, every missed pair in a cluster of at least $\epsilon_{\mathrm{nbs}}k^{\eta/2}$ relevant pairs is hit with high probability; the remaining pairs contribute at most $2n\epsilon_{\mathrm{nbs}}$. Taking $\epsilon_{\mathrm{nbs}}=\Theta(\lambda^4)$ gives cost $\widetilde O(w_{\max}^2k^{2-\eta/2}/\lambda^6)$.

\paragraph{Small $\lambda$ and the exponent.}
Below a crossover $\kappa$, RSSS19 samples $A$ at rate $\lambda^3$ and greedily extends against $B$ in $\widetilde O(n^2\lambda^7)$ time. With $d=\sqrt n\lambda$, $\gamma=2/3$, and $\kappa=n^{-1/140}$, the two Step~1 exponents are $155/84\approx1.8452$ and $389/210\approx1.8524$, the Step~2 exponent is $1589/840\approx1.8917$, and the small-$\lambda$ exponent is $39/20$; hence the maximum is $39/20$. These numbers reproduce Theorem~2.1 of \cite{rsss19}, while Section~\ref{sec:tech-ours} removes the exact-size multiplicity above.

\begin{figure}[p]
\centering
\includegraphics[width=0.88\linewidth,height=0.25\textheight,keepaspectratio]{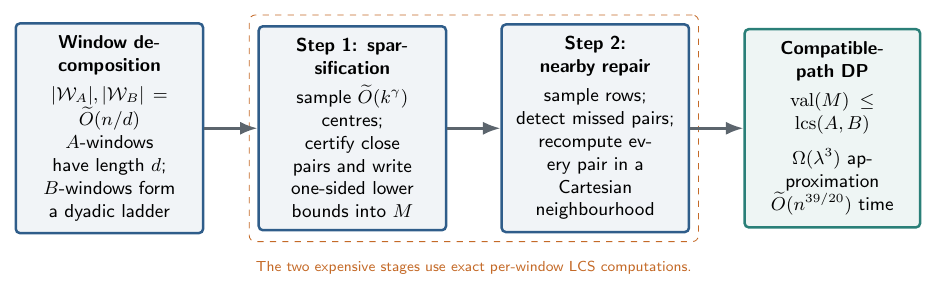}
\caption{The RSSS19 pipeline summarized in Section~\ref{sec:tech-rsss}. The two expensive middle stages compute and repair window-pair estimates before the compatible-path dynamic program.}
\label{fig:rsss_pipeline}

\vspace{1em}

\includegraphics[width=\linewidth]{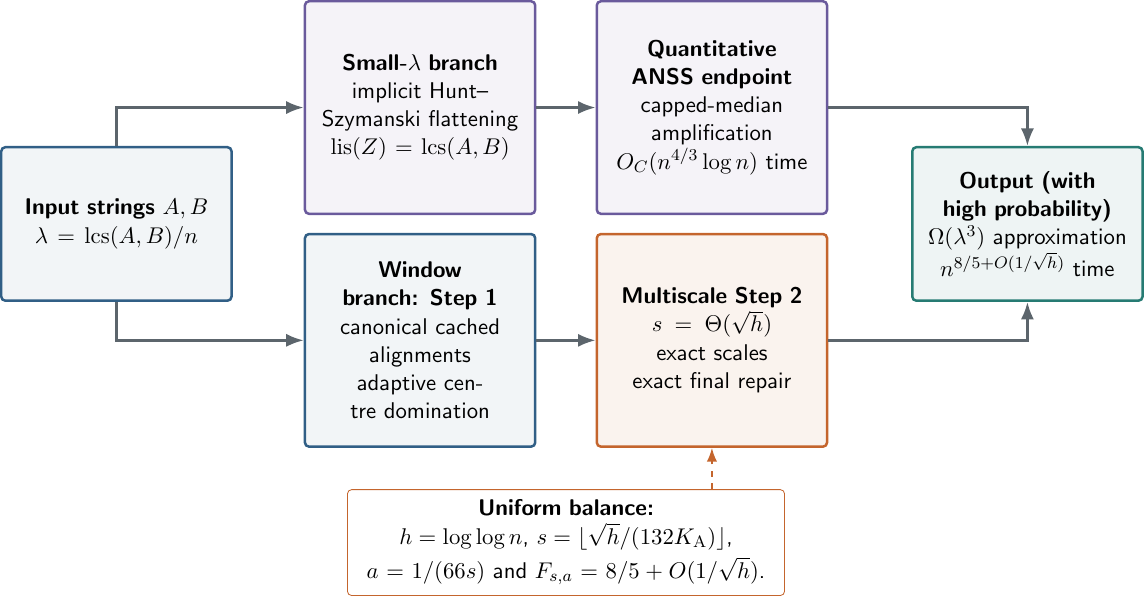}
\caption{The $8/5+o(1)$ algorithm of Theorem~\ref{thm:uniform_eight_fifths_formal}. The low-density branch uses the quantitative ANSS endpoint, while the window branch uses adaptive centre domination and $\Theta(\sqrt{\log\log n})$ scales of exact localization.}
\label{fig:our_approach_pipeline}
\end{figure}

\subsection{Summary of Our Approach}
\label{sec:tech-ours}

\paragraph{Re-balancing.} The four costs of Section~\ref{sec:tech-rsss} are not
balanced against one another: at the parameters of \cite{rsss19} they read
$n^{1.8452}$, $n^{1.8524}$, $n^{1.8917}$ and $n^{1.95}$. Treating $d$, $\gamma$ and $\kappa$ as free
and equalizing them already gives $\widetilde O(n^{151/79})$, with no change to the algorithm (Remark~\ref{rem:reproduce}).

\paragraph{Exploiting the asymmetry of the window construction.} Only $W_B$ carries
the ladder of lengths; every window of $W_A$ has length exactly $d$. Two things
follow. First, the estimate table is indexed by $W_A\times W_B$ and the dynamic
program reads nothing else, so the graph of missed pairs may be taken bipartite;
taking its maximum-degree vertex on the $B$ side puts the whole neighbourhood, and
hence every witness the Tur\'an argument produces, inside $W_A$. Centres may therefore be confined to $W_A$, where they all have length $d$ (Lemma~\ref{lem:bipartite}). Second, every LCS
the algorithm then computes has an endpoint of length $d$, so the per-pair cost is $P = \widetilde O(d^2 + w_{\max})$ rather than $w_{\max}^2$ (Lemma~\ref{lem:perpair}): the padding used in
\cite{rsss19} to equalize window lengths is an analytical device, and the dynamic
program may run unpadded.

\paragraph{Step 2.} Two corrections for the nonrecursive baseline, in which exact
detection and exact neighborhood repair pay the same per-pair cost. The entries of
$M$ are LCS values of an $A$-window against a $B$-window, hence bounded by $d$,
so an undetected pair contributes error $d$ and not $w_{\max}$. In this common-cost
model the radius $W$ is a free parameter: detection costs $\Theta(1/W)$ and
recomputation costs $\Theta(W)$, so the choice
$W = k^{\eta/2}w_{\max}$ of \cite{rsss19} sits a factor $w_{\max}/d$ off the
optimum $W = k^{\eta/2}d$. Together these replace $\lambda^{-6}$ by $\lambda^{-4}$ (Lemma~\ref{lem:step2}).
Within this additive-error calculation the power $\lambda^{-4}$ is forced:
the lost mass must remain below $\epsilon$ times a capped solution of size
$\Theta(\lambda^4 n)$. This is not a lower bound on other repair mechanisms.

\paragraph{An intersection surrogate.} The test of \cite{rsss19} asks whether
$\lcs(\opt_{i,a},w_j)$ is large. This is an LCS against a whole string, costing
$O(w_{\max})$ per pair and $k^{2+\gamma}w_{\max}$ in all. We replace it by
$I_a(i,j) := |\opt_{i,a}\cap Y_{a,j}|$, an intersection of two subsets of the
position set of the centre. It is sound: $Y_{a,j}$ is a union of $m\le4/\lambda^2$
disjoint common subsequences of $w_a$ and $w_j$, so some constituent meets
$\opt_{i,a}$ in a $1/m$ fraction of $I_a(i,j)$, and composing that with
$\opt_{i,a}$ gives $\lcs(w_i,w_j)\ge I_a(i,j)\lambda^2/4$. It is complete because
$Y_{c,b}$ is grown until the residual has LCS below $d\lambda^2/4$: splitting
$\opt_{a,c}$ across $Y_{c,b}$ yields
$|\opt_{a,c}\cap Y_{c,b}| \ge \lambda^2D_{ab}/2 - d\lambda^2/4 \ge \lambda^2D_{ab}/4$,
and the certificates the Tur\'an argument produces are exactly of this shape (Lemmas~\ref{lem:sound} and~\ref{lem:complete}).

\paragraph{One round, one matrix product: the oblivious Step~1.} A round is a set of inner
products of $0/1$ vectors over the position set of the centre. For a threshold
$\theta$ and a dyadic $B$-layer $(d_q/2,d_q]$, sample each centre coordinate
with probability
$p_q=\min\{1,2C\log n/(\theta^2d_q)\}$. After the matrix product, an eligible
entry $(i,j)$ is compared with the pair-specific threshold
$(3/16)p_q\theta^2D_{ij}$. A Chernoff bound separates
$I_a(i,j)\geq\theta^2D_{ij}/4$ from
$I_a(i,j)<\theta^2D_{ij}/8$ simultaneously for all pairs (Lemma~\ref{lem:sketch}). An accepted entry
therefore certifies
$\lcs(w_i,w_j)\geq\theta^4D_{ij}/32$, and this certified value, rather than
the raw product entry, is written into $M$. Ineligible pairs are masked, and a
fresh sample is used for each threshold, layer, and centre round.

With sufficiently small fixed crossover slack, the sampled dimension is $k^{o(1)}$ throughout the window branch. Since the
dual exponent of rectangular matrix multiplication is positive, the product of
a $k\times k^{o(1)}$ matrix by a $k^{o(1)}\times k$ matrix takes $k^{2+o(1)}$ time (Lemma~\ref{lem:round}). Rounds cannot be batched together: concatenating $g$ of them
computes $\sum_a I_a(i,j)$, and a pair may clear the threshold only in the sum.
That would not provide a one-centre certificate and could violate $M(w,w')\leq\lcs(w,w')$; batching is therefore ruled out, and the next paragraph escapes the resulting $k^{2+\gamma}$ cost by adaptivity instead. The window branch is run in the four
swap-and-reversal orientations required by the structural window lemma, and the
algorithm returns the largest of the four one-sided numerical estimates.

\paragraph{Adaptive centre domination.} The oblivious round tests every pair against every sampled centre, which is what makes Step~1 cost $k^{2+\gamma}$; but the staircase needs only that few $\theta$-high pairs remain unmarked, not a complete table. Fix a $B$-window $w_a$, write $D_a=\max\{|w_a|,d\}$, and let $S_b\subseteq[|w_a|]$ be the positions of $w_a$ used by the canonical alignment with $w_b$, so $|S_b|\ge\theta D_a$ for every high neighbour $b$. Counting incidences shows that among any $\lceil2/\theta\rceil$ high neighbours some two satisfy $|S_b\cap S_c|\ge\theta^2D_a/2$, so the graph joining such pairs has independence number below $2/\theta$ (Lemma~\ref{lem:anchor_independence}); and such an edge is visible to a single centre, because the shared positions compose into a common subsequence of $\opt_{a,c}$ and $w_b$ (Lemma~\ref{lem:edge_marks}). Hence $O(1/\theta)$ centres dominate every anchor. They cannot be identified by inspection, but $O(1/\theta)$ independent pools of $\widetilde O(k^\gamma\theta)$ centres suffice: selecting from each pool one high centre that is still unmarked either drives the unmarked set below $\widetilde O(k^{1-\gamma}/\theta)$ or exhibits an independent set that is too large (Lemma~\ref{lem:pool_domination}). A selected centre is then tested against all windows without any matrix product, by retaining $\widetilde O(\theta^{-2})$ coordinates of $\opt_{a,c}$ and walking the inverted lists of the sets $Y_{c,b}$ (Lemma~\ref{lem:sparse_filter}). This replaces $k^{2+\gamma}$ by $k^2\nu^{-3}$ (Lemma~\ref{lem:adaptive_step_one}), and the window branch of Theorem~\ref{thm:explicit_anss_free_formal} uses no rectangular matrix multiplication at all.

\paragraph{Small \texorpdfstring{$\lambda$}{lambda}.}
Section~\ref{sec:anss_endpoint} turns the ANSS sublinear-LIS estimator into
an LCS endpoint through the Hunt--Szymanski match sequence.  Its loss is
$n^{O(1/\sqrt{\log\log n})}$ and worst-case time
$O(n^{4/3}\log n)$ after capping and median amplification.  At a fixed
crossover this loss is $n^{o(1)}$; at a crossover
$n^{-\Theta(1/\sqrt{\log\log n})}$ it is absorbed by $\lambda^{-3}$.
Separately, the explicit ANSS-free
Theorem~\ref{thm:explicit_anss_free_formal} runs the
algorithm of Bringmann, Cohen-Addad and Das \cite{bcd21} with
$\varepsilon_{\mathrm{BCD}}=30/127$, giving approximation factor
$\widetilde O(n^{12/127})$ in $O(n^{224/127})$ time.

\paragraph{Multiscale exact localization.}
Algorithm~7 of \cite{rsss19} uses per-pair LCS computations both to detect an
underestimated sampled pair and to repair its Cartesian neighborhood. We keep
these computations exact but localize repeatedly. Sampled detection identifies
deficient cells of the frozen staircase target of
Lemma~\ref{lem:staircase_cap}; every trigger opens a box, and the next scale
samples afresh inside the union of those boxes. Only the final candidate region
is repaired exhaustively. Conditioning first on Step~1 and then fixing the
canonical comparator path makes the fresh row samples independent of that
path, which is the probability space used in Lemma~\ref{lem:capped_repair} and
Remark~\ref{rem:good_event}.

\paragraph{The explicit ANSS-free balance.}
With three scales, choose
$a=4/127$, $x=106/127$, and $\gamma=38/53$. At the crossover the radii are
$n^{46/127}$, $n^{36/127}$, and $n^{26/127}$. Centre preprocessing, adaptive
filtering, the four exact-localization costs, and the BCAD call are then all at
most $\widetilde O(n^{224/127})$; the block dynamic program and remaining
preprocessing are smaller. Corollary~\ref{cor:explicit_three_scale} performs
the balance, and Theorem~\ref{thm:explicit_anss_free_formal} supplies the density-oblivious
wrapper, padding, exact arithmetic, and global failure bound.

\paragraph{The $8/5$ balance.}
Lemma~\ref{lem:one_level_adaptive} gives the window cost
\[
 \frac85+\frac4{25s+10}+\frac{21s}{5s+2}a.
\]
Here $s$ is the number of localization scales and $a$ is the crossover
exponent.
Choosing fixed $s=\Theta(1/\delta)$ and fixed $a=\Theta(1/s)$ proves
Theorem~\ref{thm:fixed_delta_formal}.  Choosing
$s=\Theta(\sqrt{\log\log n})$ and
$a=\Theta(1/\sqrt{\log\log n})$, and balancing against the quantitative
ANSS endpoint, proves Theorem~\ref{thm:uniform_eight_fifths_formal}.  The
uniform factors in the multiscale lemma are only
$s^{O(1)}(\log n)^{O(1)}$, so they fit inside the same stated error term.

\section{Preliminaries}
\label{sec:preliminaries}

Section~\ref{sec:basic_notation} introduces our notation for strings, LCS densities, and compatible window paths, and specifies the approximation guarantees and computational model.
Section~\ref{sec:previous_tools} collects the results from previous work on the RSSS window framework, small-density LCS approximation, sublinear LIS estimation, extremal graph theory, and rectangular matrix multiplication.

\subsection{Basic notation}
\label{sec:basic_notation}

For a positive integer $m$, let $[m]:=\{1,\ldots,m\}$. Strings are indexed
from one. For a string $x$, we write $|x|$ for its length and $x[i..j]$ for the
substring occupying positions $i,\ldots,j$. A \emph{window} is a contiguous
substring, identified with both its interval and the string on that interval.

For strings $x$ and $y$, let $\lcs(x,y)$ denote the length of a longest common
subsequence. For a sequence $z$ over a totally ordered set, let $\lis(z)$
denote the length of a strictly increasing subsequence. The main input consists
of two strings $A,B\in\Sigma^n$. We use the standard word-RAM model with
$\Theta(\log n)$-bit words. At the start of the algorithm, the at most $2n$
input symbols are coordinate-compressed to integers in $[2n]$: if the symbols
are not already word-sized identifiers, sorting and deduplication takes
$O(n\log n)$ comparisons. All later equality tests therefore use word-sized
integers, and the private padding symbols are fresh integers outside this range.
This preprocessing is negligible for every running-time bound below. We use
\[
L:=\lcs(A,B),
\qquad
\lambda:=\frac{L}{n}=\|\lcs(A,B)\|.
\]
For a pair of windows $(w_i,w_j)$, define
\[
L_{ij}:=\lcs(w_i,w_j),
\qquad
D_{ij}:=\max\{|w_i|,|w_j|\},
\qquad
\rho_{ij}:=\frac{L_{ij}}{D_{ij}}.
\]
Thus $\lambda$ is the global LCS density and $\rho_{ij}$ is the normalized
density of a window pair.

A numerical estimate $\widehat L$ is a \emph{one-sided $q$-approximation},
where $q\in(0,1]$, if
\[
q\lcs(A,B)\le \widehat L\le \lcs(A,B).
\]
Equivalently, an approximation factor $Q\ge1$ means
$\lcs(A,B)/Q\le\widehat L\le\lcs(A,B)$. In particular, an
$\Omega(\lambda^3)$-approximation is one-sided unless stated otherwise.
This paper proves numerical estimates; it does not claim that the final
algorithm outputs a common-subsequence witness.

For two window families $W_A$ and $W_B$, a \emph{compatible window path} is a
sequence of pairs from $W_A\times W_B$ that can be ordered so that the windows
are pairwise disjoint and occur from left to right on both strings, as in
\cite{rsss19}. For a nonnegative table $M$ indexed by $W_A\times W_B$, let
\[
\operatorname{val}(M)
:=
\max_{P}\sum_{(i,j)\in P}M_{ij},
\]
where the maximum is over compatible window paths. We call $M$ one-sided if
$0\le M_{ij}\le L_{ij}$ for every entry. Compatibility then implies
$\operatorname{val}(M)\le\lcs(A,B)$.

The notation $\widetilde O(\cdot)$ suppresses factors polylogarithmic in the
input length. Constants may depend on fixed accuracy, slack, and failure
parameters. With high probability means failure probability at most $n^{-C}$
for any fixed requested constant $C>0$, after amplification. All random choices
are independent unless explicitly coupled.

\subsection{Tools from previous work}
\label{sec:previous_tools}

\paragraph{The RSSS window framework.}
Rubinstein, Seddighin, Song, and Sun introduced the window decomposition,
close-pair sparsification, compatible-path dynamic program, and nearby-repair
architecture on which our algorithm is based \cite{rsss19}. Their complete
result is the following.

\begin{lemma}[\cite{rsss19}, Theorem 2.1]
\label{lem:rsss}
There is a randomized algorithm that, given strings $A,B$ of length $n$ with
$\|\lcs(A,B)\|=\lambda$, computes an $\Omega(\lambda^3)$-approximation of
$\lcs(A,B)$ in time $\widetilde O(n^{39/20})$.
\end{lemma}

\begin{definition}[Window-compatible common subsequence; \cite{rsss19}, Definition~4.1]
\label{def:window_compatible}
Let $\mathcal W_A$ and $\mathcal W_B$ be families of windows of $A$ and
$B$. Let $\mathcal S=\langle w_1,\ldots,w_\alpha\rangle$ and
$\mathcal S'=\langle w'_1,\ldots,w'_\alpha\rangle$ be sequences of
pairwise non-overlapping windows from $\mathcal W_A$ and $\mathcal W_B$,
respectively, ordered by their starting positions. A common subsequence of $A$
and $B$ is \emph{window-compatible with respect to}
$(\mathcal S,\mathcal S')$ if it is the concatenation, in this order, of a
common subsequence of $(w_i,w'_i)$ for every $i\in[\alpha]$. It is
\emph{window-compatible} if such sequences $\mathcal S,\mathcal S'$ exist.
\end{definition}

\begin{definition}[Window construction for arbitrary $\lambda$; \cite{rsss19}, Definition~4.5]
\label{def:rsss_windows}
Fix a base window length $d$ and a parameter $\epsilon_0\in(0,1)$. Put
$f:=\log_2(1/\epsilon_0)$, $d_i:=d2^i$, and $t_i:=n/d_i$; as in
\cite{rsss19}, we suppress immaterial integer rounding. The $A$-windows and
the zeroth layer of $B$-windows are
\[
\mathcal W_A:=\{A[yd+1..(y+1)d]:0\leq y<t_0\},
\qquad
\mathcal W_B^0:=\{B[yd+1..(y+1)d]:0\leq y<t_0\}.
\]
For every $i\in\{1,\ldots,f\}$, define
\[
\mathcal W_B^i:=\{B[yd_i+1..yd_i+d_i/2+xd]:
1\leq x\leq d_i/(2d),\ 0\leq y<t_i\}.
\]
Finally, set $\mathcal W_B:=\mathop{\cup}_{i=0}^f\mathcal W_B^i$. Thus the
windows in layer $i$ start at multiples of $d_i$, and their lengths are
multiples of $d$ lying in $(d_i/2,d_i]$.
\end{definition}

\paragraph{Rounding, padding, and orientation convention.}
Algorithm~\ref{alg:approximate_lcs} fixes a common padded pair before taking any
orientation. More precisely, its base length $d$, its reciprocal-power-of-two
parameter $\epsilon_{\mathrm{win}}$, and its dyadic guess grid make every largest
layer scale $d_f=d/(\epsilon_{\mathrm{win}}\nu)$ a power of two. Let
$d_{\max}$ be the largest of these scales. The algorithm right-pads the original
two input strings once, using private padding alphabets disjoint from each other
and from the input alphabet, to the least common length
$n^+<n+d_{\max}\le2n$ divisible by $d_{\max}$. It then forms the four
swap/reversal orientations of this one padded pair. Thus a reversed call has
the padding on the left, exactly as the reversal of the common padded string
requires. Since every $d_f$ divides $d_{\max}$, all $f,d_i,t_i$, window starts,
and allowed $B$-window lengths are integers for every guess, and there is no
final partial $A$-window. The LCS is unchanged because no padding symbols match.
The structural lemma is applied to the common padded pair, so its four
orientations are exactly the four calls made by the algorithm.

\begin{lemma}[Window-compatible structural lemma;
 \cite{rsss19}, Lemma~4.7]
\label{lem:rsss_window_compatible}
There is an absolute constant $c_{\mathrm{str}}>0$ with the following
property. Let $X,Y$ be strings of length $n$, and let $\epsilon_0>0$. For at
least one of the four oriented pairs
\[
(A,B)\in\{(X,Y),(Y,X),(\operatorname{reverse}(X),\operatorname{reverse}(Y)),
(\operatorname{reverse}(Y),\operatorname{reverse}(X))\},
\]
the following holds. If $\|\lcs(A,B)\|\geq\lambda$, then the window
sets from Definition~\ref{def:rsss_windows} contain a window-compatible
common subsequence of $A$ and $B$ of length at least
$c_{\mathrm{str}}\lambda n-2\epsilon_0 n$.
\end{lemma}

We also use
the numerical block dynamic program and
nearby-repair organization in their Algorithms 6 and 7. Whenever our argument
changes one of these interfaces, the modified statement is proved explicitly.

\paragraph{The small-density approximation.}
Bringmann, Cohen-Addad, and Das give the tradeoff used in the first
small-$\lambda$ branch.

\begin{lemma}[\cite{bcd21}, Theorem~1]
\label{lem:bcd}
For every $0<\varepsilon\le1$ there is a randomized algorithm that, given two
strings of length $n$, returns a one-sided estimate with approximation factor
$\widetilde O(n^{2\varepsilon/5})$ in time $O(n^{2-\varepsilon})$ with high
probability.
\end{lemma}

\paragraph{Sublinear LIS estimation.}
The endpoint argument uses the following ordinary-LIS lemma of
Andoni, Nosatzki, Sinha, and Stein.

\begin{lemma}[\cite{anss22}, Theorems~1.1 and~4.3]
\label{lem:anss}
Let $1/N<\mu<1$, suppose $\mu=o(1)$, and put
$\epsilon=1/\log\log(1/\mu)$.  There is a randomized non-adaptive algorithm
that, given random access to a sequence $y$ of length $N$ satisfying
$\lis(y)\ge\mu N$, has a joint success event on which its output satisfies
\[
\frac{\lis(y)}{\alpha}-\mu N\le\widehat L\le\lis(y),
\qquad
\alpha=(1/\mu)^{\sqrt\epsilon}
(\log(1/\mu))^{2^{O(\log^2(1/\epsilon)/\sqrt\epsilon)}}.
\]
The event holds with high probability, and the total expected running time of
the complete call is $(1/\mu)N^{o(1)}$ at this parameter schedule.
\end{lemma}

Here and throughout, ``increasing'' means strictly increasing. This agrees
with \cite[Definition~2.1]{anss22}, which uses the strict order on
$\mathbb N$.

The two inequalities in Lemma~\ref{lem:anss} are asserted on the same event;
we do not treat the upper inequality as deterministic.  Section~
\ref{sec:anss_endpoint} supplies the random-access implementation, converts
the expected time to a worst-case cap, and amplifies only a constant one-run
success probability by independent repetitions and a median.

\paragraph{Tur\'an's theorem.}
We use the following equivalent form of the classical extremal bound.

\begin{lemma}[\cite{turan41}]
\label{lem:turan_form}
Let $H$ be a graph on $q$ vertices. If the complement of $H$ contains no clique
of size $r$, where $r\ge2$, then
\[
|E(H)|\ge \frac{q^2}{2(r-1)}-\frac q2.
\]
\end{lemma}

\paragraph{Rectangular matrix multiplication.}
Let $\alpha$ denote the dual exponent of matrix multiplication. The current
bound $\alpha>0.321$ \cite{wxxz24} implies the following consequence, which is
all that our proof uses.

\begin{lemma}[\cite{wxxz24}]
\label{lem:rect_mm}
For every fixed $\alpha_0<\alpha$, a $k\times s$ matrix can be multiplied by an
$s\times k$ matrix in $k^{2+o(1)}$ field operations whenever
$s\le k^{\alpha_0}$. Over fields represented by $O(\log n)$-bit words, the
word-RAM overhead is polylogarithmic.
\end{lemma}

\section{An Improved Running Time for the RSSS LCS Algorithm}
\label{sec:proof}

Section~\ref{sec:accounting} sets up the window parameters and running-time costs.
Section~\ref{sec:bipartite} restricts sparsification centres to the $A$-side and bounds the number of missed pairs.
Section~\ref{sec:perpair} gives a faster exact LCS computation for each window pair.
Section~\ref{sec:step2} optimizes the nearby-search radius to reduce the repair cost.
Section~\ref{sec:surrogate} replaces the sparsification test by an intersection surrogate and proves its soundness and completeness.
Section~\ref{sec:mm} sketches these intersections and evaluates the tests in batches using rectangular matrix multiplication.
Section~\ref{sec:small-lambda} applies the small-density LCS approximation of Bringmann, Cohen-Addad, and Das.
Section~\ref{sec:anss_endpoint} converts the ANSS LIS estimator into an LCS approximation with a worst-case time bound and amplified success probability.
Section~\ref{sec:optimization} balances the initial costs to obtain an explicit nonrecursive bound.
Section~\ref{sec:adaptive_domination} develops adaptive centre selection and sparse filtering to accelerate sparsification.
Section~\ref{sec:bootstrap} combines the local-density thresholds into a global staircase target and bounds the set of deficient entries.
Section~\ref{sec:capped-repair} repairs this set through successive rounds of exact detection and localization.
Section~\ref{sec:iterate} balances the multiscale costs and derives both a uniform exponent and an explicit three-scale bound.
Finally, Section~\ref{sec:main-proof} assembles the full algorithms, handles rounding and failure probabilities, and proves the stated approximation and running-time guarantees.

\subsection{The cost of the window framework}
\label{sec:accounting}

We first record the accounting in a form that exposes the free parameters. Fix a
base window length $d$. The construction of \cite[Definition 4.5]{rsss19} produces
a set $W_A$ of $A$-side windows and a set $W_B$ of $B$-side windows with the
following properties \cite[Algorithm~5 and Fact~4.6]{rsss19}: write $k_A:=|W_A|=n/d$, $k_B:=|W_B|$, and $k_{\mathrm{tot}}:=k_A+k_B$.
Then $k_B,k_{\mathrm{tot}}=\widetilde O(n/d)$; the maximum window length is
$w_{\max} = \Theta(d/\lambda)$; and the minimum window length is
$w_{\min} = \Theta(d)$.  The $B$-windows are organized into
$O(\log(1/\lambda))$ dyadic layers.  Within one such layer their exact lengths
are not equal, but all lie within a factor two; throughout the proof a
\emph{layer} means this dyadic bucket, and every threshold involving a window
pair uses its actual length $D_{ij}$.  The additive loss incurred by the
decomposition is $2\epsilon_0 n$ with $\epsilon_0 = \epsilon\lambda$
\cite[Lemma~4.7]{rsss19}, independent of $d$.

\begin{observation}[Asymmetry of the window construction]
\label{obs:asymmetry}
The window construction is asymmetric: every window of $W_A$ has length exactly
$d$, and only $W_B$ carries the dyadic layers.  In a layer with upper scale
$d_i$, its exact lengths lie in $(d_i/2,d_i]$.
\end{observation}

\begin{proof}
Immediate from the construction of \cite[Algorithm 5]{rsss19}: the $A$-side loop
emits only windows $A(\mathrm{left}, d)$, whereas the $B$-side loop emits, for each
layer $i$ with $d_i = d\cdot 2^i$, windows of every length
$d_i/2 + x\cdot d$ for $x = 1, \dots, d_i/(2d)$. The same is said in the prose of
\cite[Section 4]{rsss19}: ``$W_A$ will simply be a partitioning of $A$ into disjoint
windows of length $d$''.

\end{proof}

Two consequences we use repeatedly. First, every $B$-window has length at least
$d$, since $W_B^0$ consists of length-$d$ windows and layer $i\ge1$ contributes
lengths in $(d_i/2, d_i]$ with $d_i \ge 2d$. Hence for every pair consulted by the
algorithm $D_{ij} \triangleq \max\{|w_i|,|w_j|\} = |w_j|$, the $B$-side length, and
$D_{ij}\ge d$; in particular the case distinction ``whether the $A$ side window is
larger or the $B$ side window is larger'' in \cite[proof of Theorem 2.1]{rsss19} is
vacuous here. Second, $W_A$ is a partition of $A$ into $n/d$ windows spaced $d$
apart, which is what Observation~\ref{obs:nearby-counts} counts.

Observation~\ref{obs:asymmetry} is the source of two of our savings. Note also the
consequence
\begin{align}
\label{eq:total-length}
\sum_{w \in W_A \cup W_B} |w|  =  \widetilde O(n/\lambda),
\end{align}
since each of the $O(\log (1/\lambda))$ layers contributes $O(n/d)$ windows and the
layer lengths $d_i$ sum to $O(d/\lambda)$.

For exponent bookkeeping we write $k=n^x$ for the common scale $n/d$,
suppressing polylogarithmic factors; whenever an exact cardinality is needed
we retain $k_A$ and $k_B$ explicitly. Thus $d=n^{1-x}$, and
$\lambda=n^{-a}$, hence
$w_{\max} = n^{1-x+a}$. Let $\gamma \in (0,1)$ be the sparsification parameter of
\cite[Theorem 3.21]{rsss19}, let $P$ denote the cost of one window-pair LCS
computation, and let $Q$ denote the cost of testing all pairs in one sparsification
round. The four costs of the algorithm are then
\begin{align}
T_1 &= k^{1+\gamma}\cdot P
&&\text{(Step 1: preprocessing)},
\label{eq:T1}\\
T_2 &= k^{\gamma}\cdot Q
&&\text{(Step 1: pair tests)},
\label{eq:T2}\\
T_3 &= P\cdot k^{2-\eta/2}\lambda^{-q}
&&\text{(Step 2: nearby search)},
\label{eq:T3}\\
T_0 &= n^{2-ra}
&&\text{(the small-$\lambda$ routine)},
\label{eq:T0}
\end{align}
where $k^{2-\eta}$ bounds the number of pairs underestimated by Step~1 and
$r$ measures the quality of the routine used when $\lambda$ is small. Following
\cite[proof of Theorem 2.1]{rsss19}, Step~1 is run once for each threshold
$\lambda'$ in a multiplicative net of $O(\epsilon^{-1}\log(1/\lambda))$ values and
once for each dyadic layer pair.  There are only polylogarithmically many such
runs.  We never split a dyadic layer into its many exact lengths; instead the
matrix-product output is compared entrywise with the actual $D_{ij}$-dependent
threshold.  These factors, as well as the outer geometric search over
$\lambda$, are suppressed throughout. In
\cite{rsss19} one has $P = w_{\max}^2$, $Q = k^2 w_{\max}$, $q = 6$, and the
routine of Eq.~\eqref{eq:T0} is sample-and-extend with $r = 7$.

\begin{remark}
\label{rem:reproduce}
Substituting the parameters of \cite[proof of Theorem 2.1]{rsss19}, namely
$d = \sqrt n\lambda$, $\gamma = 2/3$ and $\lambda = n^{-1/140}$, into
Eq.~\eqref{eq:T1}--\eqref{eq:T3} reproduces the exponents $1.8452$, $1.8524$ and
$1.8917$ stated there, and Eq.~\eqref{eq:T0} gives $1.95$. The three window terms are
therefore not balanced against the small-$\lambda$ term in \cite{rsss19}; already
re-optimizing $d$, $\gamma$ and the crossover in Eq.~\eqref{eq:T1}--\eqref{eq:T0}
lowers the exponent from $39/20$ to $151/79 = 1.9114\ldots$ with no change to the
algorithm. That $d$ is free to be re-optimized is noted in \cite[Section 4]{rsss19}
itself: ``one can play with the value of $d$ to optimize the running time''.
\end{remark}

\subsection{Restricting the sparsification to \texorpdfstring{$A$}{A}-side centres}
\label{sec:bipartite}

Step~1 of \cite{rsss19} repeats $\widetilde O(k^\gamma)$ rounds; in a round with
centre $w_c$ it builds the data structure of \cite[Lemma 3.22]{rsss19} and tests all
pairs. Two quantities control the cost: the length of the centre, which enters $P$
quadratically, and the number of pairs tested. In \cite{rsss19} the centre is drawn
from all of $W_A \cup W_B$, so it may have length $w_{\max}$. We show the centres
may be confined to $W_A$, where every window has length $d$.

Recall the notation of \cite[Section 3.2]{rsss19}: $\opt_{i,a}$ is a fixed longest
common subsequence of $w_i$ and $w_a$, viewed as a set of positions of $w_a$, and
$\lcs_{w_a}(w_i,w_j)=\lcs(\opt_{i,a},w_j)$. We make this choice algorithmic:
order the two windows by their global interval identifiers, choose the
lexicographically smallest sequence of matched position pairs among all maximum
alignments, cache it under that unordered pair, and use its two coordinate
projections in both orientations. Thus $\opt_{i,a}$ and $\opt_{a,i}$ always
come from the same cached alignment; Lemma~\ref{lem:edge_marks} uses this
convention.

This notion identifies mixed window pairs that have many $A$-side centres
providing a substantial common-subsequence certificate. Keeping the certificate
in absolute units accommodates the unequal window lengths used below.

\begin{definition}[$A$-close pairs]
\label{def:aclose}
For a mixed pair $(w_i,w_j)$ put $D_{ij}:=\max\{|w_i|,|w_j|\}$.  The pair is
\emph{$A$-close} if there are at least $k_A^{1-\gamma}$ windows $w_c\in W_A$
such that
\begin{align*}
\lcs_{w_c}(w_i,w_j)\ge \lambda^2D_{ij}/2,
\end{align*}
where $k_A=|W_A|$.  This is an absolute-length condition.  Equivalently, after
padding the shorter windows by nonmatching dummy characters as in
\cite[proof of Theorem 3.21]{rsss19}, it is the normalized condition
$\|\lcs_{w_c}(w_i,w_j)\|\ge\lambda^2/2$.
\end{definition}

The standard dynamic program implements this tie-breaking by storing, together
with each optimum length, the lexicographically first predecessor; reconstructing
and caching the resulting alignment changes neither the asymptotic time nor the
space bound.

The next lemma adapts the missed-edge counting argument to the asymmetric,
bipartite window family. It shows that detecting all $A$-close pairs leaves only
a sparse set of relevant false negatives.

\begin{lemma}[Bipartite false-negative bound]
\label{lem:bipartite}
Let $\mathrm{NG}^{\mathrm{bip}}_\lambda$ be the bipartite graph on $W_A \sqcup W_B$
whose edges are the pairs $(i,j) \in W_A\times W_B$ with
$\lcs(w_i,w_j)\ge\lambda D_{ij}$ that Step~1 fails to mark. If every $A$-close pair is
marked, then $|E(\mathrm{NG}^{\mathrm{bip}}_\lambda)| = \widetilde O(k^{2-\gamma}/\lambda)$.
\end{lemma}

\begin{proof}
Restricting attention to $W_A\times W_B$ is legitimate because the block dynamic
program of \cite[Algorithm 6]{rsss19} reads the estimate matrix only at entries
indexed by an $A$-window and a $B$-window; false negatives on $A\times A$ or
$B\times B$ pairs are never consulted.

Suppose, for contradiction, that
$|E(\mathrm{NG}^{\mathrm{bip}}_\lambda)| > 4k_A^{1-\gamma}k_B/\lambda$, where
$k_B = |W_B|$. Since the edge set is bipartite,
$|E| = \sum_{b\in W_B}\deg(b) \le k_B\max_{b\in W_B}\deg(b)$, so some
$w_a \in W_B$ has $q:=\deg(w_a) > 4k_A^{1-\gamma}/\lambda$. Put
$D:=|w_a|$. The neighbourhood $N(w_a)$ is contained in $W_A$, so all its
members have length $d\le D$. Define an auxiliary graph $H_a$ on $N(w_a)$ by
joining $w_i$ and $w_j$ when both
\begin{align}
\label{eq:absolute-nf}
\lcs_{w_a}(w_i,w_j)\ge\lambda^2D/2
\quad\text{and}\quad
\lcs_{w_a}(w_j,w_i)\ge\lambda^2D/2.
\end{align}
This is the absolute-unit version of the auxiliary graph of
\cite[Definition 3.20]{rsss19} after the shorter windows are padded to length $D$.

We verify from first principles that every set $T\subseteq N(w_a)$ of size
$t\ge r:=\lceil2/\lambda\rceil$ contains an edge of $H_a$. For $w_i\in T$, let
$S_i:=\opt_{i,a}\subseteq[D]$, and put
$m(x):=|\{i\in T:x\in S_i\}|$. Since $(w_i,w_a)$ is an edge of the missed-pair
graph, $|S_i|=\lcs(w_i,w_a)\ge\lambda D$. Thus
\[
\sum_{i\ne j}|S_i\cap S_j|=\sum_{x=1}^{D}m(x)(m(x)-1)\ge \frac{(\sum_xm(x))^2}{D}-\sum_xm(x)\ge t\lambda D(t\lambda-1).
\]
The first identity counts ordered incidences and the first inequality is
Cauchy--Schwarz. Dividing by at most $t^2$ ordered pairs, some $i\ne j$ satisfies
\[
|S_i\cap S_j|\ge\lambda D(\lambda-1/t)\ge\lambda^2D/2.
\]
The common positions in $w_a$ compose monotonically through the two fixed LCS
matchings, in either orientation, so Eq.~\eqref{eq:absolute-nf} holds.

Consequently the complement of $H_a$ is $K_r$-free. Lemma~\ref{lem:turan_form} gives
\[
|E(H_a)|\ge \frac{q^2}{2(r-1)}-\frac q2>\frac{q^2\lambda}{4}-\frac q2.
\]
Hence some $w_b\in N(w_a)$ has degree greater than $q\lambda/2-1$, which is at
least $k_A^{1-\gamma}$ because $q>4k_A^{1-\gamma}/\lambda$.

For every $w_c$ adjacent to $w_b$ in $H_a$ we have, in raw units,
\begin{align*}
\lcs_{w_c}(w_a,w_b)=\lcs_{w_a}(w_c,w_b)\ge\lambda^2D/2.
\end{align*}
The equality holds because $\opt_{a,c}$ and $\opt_{c,a}$ realize the same matched
string.  Since $D=D_{ab}$, all such $w_c$ lie in $N(w_a)\subseteq W_A$. Hence the pair $(w_a,w_b)$,
which lies in $W_B\times W_A$, is $A$-close and is therefore marked; but $w_b$ is a
neighbour of $w_a$ in $\mathrm{NG}^{\mathrm{bip}}_\lambda$, i.e.\ the pair is
unmarked, a contradiction. Finally $k_A, k_B = \widetilde O(k)$.
\end{proof}

Since a pair is $A$-close as soon
as $k_A^{1-\gamma}$ centres in $W_A$ witness it, sampling
$\widetilde O(k_A^\gamma) = \widetilde O(k^\gamma)$ centres uniformly from $W_A$
detects every $A$-close pair with high probability, exactly as in
\cite[Lemma 3.17]{rsss19}.  Notice that no unequal-length form of
\cite[Lemma 3.18]{rsss19} was used: the displayed incidence calculation proves the
needed statement directly in the anchor's $D$ positions.  Keeping the certificate
in absolute units is essential; normalizing Eq.~\eqref{eq:absolute-nf} by the short
length $d$ and later interpreting it with denominator $D$ would be invalid.

\subsection{The per-pair cost of a window LCS}
\label{sec:perpair}

The asymmetric construction makes one input of every window comparison short.
This permits an output-sensitive dynamic program whose cost depends quadratically
on $d$ and only linearly on the longer endpoint.

\begin{lemma}[Per-pair cost of a window LCS]
\label{lem:perpair}
With centres confined to $W_A$, every LCS computation performed by Step~1 and
Step~2 has one endpoint of length $d$, and each costs
$P = \widetilde O(d^2 + w_{\max})$.
\end{lemma}

\begin{proof}
By Observation~\ref{obs:asymmetry} the centre and one endpoint of every consulted
pair have length exactly $d$. For two strings of lengths $d \le D$ whose LCS has
length $L \le d$, preprocess the occurrence lists of every symbol in the longer
string.  After reading a prefix of the shorter string, maintain $p_\ell$, the
smallest position in the longer string at which a common subsequence of length
$\ell$ can end, with $p_0=0$.  When the next symbol is $\sigma$, update in decreasing
order of $\ell$ by
\begin{align*}
p_\ell\gets\min\{p_\ell,\operatorname{next}_\sigma(p_{\ell-1})\},
\end{align*}
where $\operatorname{next}_\sigma(t)$ is the first occurrence of $\sigma$ after
position $t$.  The decreasing order preserves the preceding row of the dynamic
program.  There are at most $L+1$ finite states in any row, so binary search in the
occurrence lists gives total time $O(D+dL\log D)=\widetilde O(D+d^2)$.
Storing the predecessor that realizes each update recovers the two position
sets of one fixed LCS alignment within the same bound, as required by Step~1.
The proof of \cite[Lemma~3.22]{rsss19}, together with
\cite[Theorem~A.8]{rsss19}, shows that the construction of $Y_{a,i}$ costs
$O(|w_a|\cdot|Y_{a,i}|\log n + |w_i|\log n)$, and
$|Y_{a,i}| \le |w_a| = d$, so it too is $\widetilde O(d^2 + |w_i|)$. Summing over a
round and using Eq.~\eqref{eq:total-length}, and noting that both $\opt_{\cdot,c}$ and
$Y_{c,\cdot}$ are computed for every window, which only doubles the count, the
total preprocessing per round is
$\widetilde O(k d^2 + n/\lambda) = \widetilde O(nd + n/\lambda)$.
\end{proof}

Note that \cite{rsss19} charges $w_{\max}^2$ per pair, which at the optimum is a
factor $\lambda^{-2}$ larger. The padding by dummy characters used in
\cite[Theorem 3.21]{rsss19} to equalize window lengths is an analytical device
only: dummy characters match nothing, so the LCS values are unchanged and the
dynamic program may be run on the unpadded strings.

\subsection{Step 2 with an optimized nearby radius}
\label{sec:step2}

Step~2 is the nearby search of \cite[Algorithm 7]{rsss19}, inherited from
\cite{cdgks18}. It declares two same-side windows \emph{nearby} when their starting
indices differ by at most a radius $W$, samples $A$-side windows at rate
$p = \widetilde O(k^{-\eta/2}/\epsilon_{\mathrm{nbs}})$, brute-forces the sampled
rows, and recomputes the LCS of every nearby pair of every detected underestimated
pair. Both factors $\lambda^{-1}$ in the bound of \cite[Theorem 5.2]{rsss19} have
the same origin.

\begin{observation}[Nearby window counts]
\label{obs:nearby-counts}
With radius $W$, the number of $A$-side windows nearby a given one is
$O(W/d+1)$.  In a $B$-side layer with upper scale $d_i$, the number nearby a
given one is $O(W/d+d_i/d)$.  Consequently, whenever $W\ge w_{\max}$, the
number of nearby windows on either side is $\widetilde O(W/d)$.
\end{observation}

\begin{proof}
$W_A$ consists of disjoint windows of length $d$ spaced $d$ apart, giving
$O(W/d+1)$.  In layer $i$ of $W_B$ the starting indices are multiples of
$d_i$, and each start carries $d_i/(2d)$ lengths.  There are
$O(W/d_i+1)$ eligible starts, hence
\[
O((W/d_i+1)d_i/d)=O(W/d+d_i/d)
\]
nearby windows in that layer.  Summing over the dyadic layers and using
$W\ge w_{\max}\ge d_i$ gives the asserted
$\widetilde O(W/d)$ bound.
\end{proof}

The nearby-pair repair step has a tunable search radius that can be balanced
against its sampling cost. Optimizing that balance removes two powers of
$\lambda^{-1}$ from the earlier analysis within the stated cost model.

\begin{lemma}[Step~2 with an optimized radius]
\label{lem:step2}
Assume $w_{\max}\le W\le n$ and that the uncapped sampling expression
below is at most one. Step~2 can be implemented in time
$\widetilde O(k^{2-\eta/2}P\lambda^{-4})$, that is, with $q=4$ in
Eq.~\eqref{eq:T3}. Within the two-term, no-reuse cost model displayed in the
proof, balancing those two charged terms cannot give a smaller power of
$\lambda^{-1}$.
\end{lemma}

\begin{proof}
Two changes to the analysis of \cite[Theorem 5.2]{rsss19}.

\emph{(i) The error of an undetected pair is at most $d$, not $w_{\max}$.} The
estimate matrix is indexed by $W_A\times W_B$, and by
Observation~\ref{obs:asymmetry} every $A$-window has length $d$; hence every entry
is an LCS value bounded by $d$.

\emph{(ii) The radius $W$ is a free parameter.} The counting argument of
\cite[Lemma 5.1]{rsss19} bounds the total error by
$2 d(\epsilon_{\mathrm{nbs}}k^{\eta/2}) n/W$, which must be at most $\epsilon$
times the size $\Theta(\lambda^4 n)$ of the solution being produced; hence
$\epsilon_{\mathrm{nbs}} = \Theta(\epsilon\lambda^4 W/(dk^{\eta/2}))$. With
$p = \widetilde O(k^{-\eta/2}/\epsilon_{\mathrm{nbs}})$ the two costs are
\begin{align*}
\text{detection} &= \widetilde O(pk\cdot k\cdot P)
= \widetilde O(k^2 P d/(\epsilon\lambda^4 W)),\\
\text{recomputation} &= \widetilde O(k^{2-\eta}p\cdot (W/d)^2\cdot P)
= \widetilde O(k^{2-\eta}WP/(\epsilon\lambda^4 d)),
\end{align*}
using Observation~\ref{obs:nearby-counts}. The first decreases in $W$ and the
second increases in $W$; they balance at $W^2 = k^\eta d^2$, that is
$W = k^{\eta/2}d$, where both equal $\widetilde O(k^{2-\eta/2}P\lambda^{-4})$. The
choice $W = k^{\eta/2}w_{\max}$ of \cite{rsss19} is a factor
$w_{\max}/d = \Theta(1/(\epsilon\lambda))$ off this optimum, which together with
(i) accounts for the two lost factors of $\lambda$.

The recomputation bound above is initially an expectation over the sampled rows.
As in \cite[proof of Theorem 5.2]{rsss19}, abort a trial once it exceeds a constant
multiple of this budget, repeat independently, and keep the largest completed
one-sided estimate. Markov's inequality gives constant completion probability per
trial, and $O(\log n)$ repetitions make both completion and correctness hold with
high probability. This adds only a polylogarithmic factor.

For clarity, the last assertion is only a statement about this displayed cost
model. If one charges both formulas, does not reuse overlapping exact
computations, and does not exploit early stopping or deduplication beyond the
bounds above, the error constraint fixes
$\epsilon_{\mathrm{nbs}}=O(\lambda^4W/(dk^{\eta/2}))$ and balancing the two
resulting terms chooses $W=k^{\eta/2}d$. This is not a lower bound for all
exact-repair implementations.
\end{proof}

Combining Lemma~\ref{lem:bipartite}, which gives
$k^{2-\eta} = \widetilde O(k^{2-\gamma}/\lambda)$ and hence
$k^{-\eta/2} = k^{-\gamma/2}\lambda^{-1/2}$, with Lemma~\ref{lem:step2} yields
\begin{align}
\label{eq:T3-final}
T_3  =  \widetilde O(P\cdot k^{2-\gamma/2}\lambda^{-9/2}),
\end{align}
so the effective exponent in Eq.~\eqref{eq:T3} is $q' = 9/2$.

\subsection{An intersection surrogate for the sparsification test}
\label{sec:surrogate}

The test of \cite{rsss19} asks whether $\lcs(\opt_{i,a}, w_j)$ is large. This is an
LCS against the whole string $w_j$, is answered in $O(w_{\max})$ time per pair, and
costs $k^{2+\gamma}w_{\max}$ in total. We replace it by a test on two subsets of
the position set of the centre, which will let us batch a whole round into one
matrix product.

The following surrogate replaces an LCS query against an entire window by an
intersection within the centre's coordinate set. Its indicator-vector form is
what makes the subsequent batched matrix multiplication possible.

\begin{definition}[Intersection surrogate]
\label{def:surrogate}
For a centre $w_a$ put $I_a(i,j) = |\opt_{i,a}\cap Y_{a,j}|$, the intersection
being taken inside the position set of $w_a$. Recall $D_{ij}$ from
Definition~\ref{def:aclose}.  Within a dyadic layer it may vary by a factor
smaller than two, and every comparison below uses the actual $D_{ij}$.
Step~1 marks the pair $(i,j)$ when
$I_a(i,j) \ge \lambda^2 D_{ij}/4$.
\end{definition}

The surrogate is used only when its intersection count certifies a genuine
common subsequence. The next lemma converts that count into a one-sided lower
bound, ensuring that marking cannot overestimate the window LCS.

\begin{lemma}[Soundness]
\label{lem:sound}
If $I_a(i,j)\ge T$ then $\lcs(w_i,w_j) \ge T\lambda^2/4$. In particular the marking
rule of Definition~\ref{def:surrogate} certifies
$\lcs(w_i,w_j)/D_{ij}\ge\lambda^4/16$, which is exactly the guarantee of
\cite[Theorem 3.21]{rsss19}.
\end{lemma}

\begin{proof}
By construction $Y_{a,j} = L_1\cup\dots\cup L_m$ is a disjoint union of common
subsequences of $w_a$ and $w_j$, each of size at least $|w_a|\lambda^2/4$; since
the $L_t$ are disjoint subsets of the position set of $w_a$ we get
$m\le 4/\lambda^2$. Let $S = \opt_{i,a}\cap Y_{a,j}$ and pick $t$ maximizing
$|S\cap L_t| \ge |S|/m \ge T\lambda^2/4$. Each position of $S\cap L_t$ is matched
monotonically to a position of $w_i$ by $\opt_{i,a}$ and monotonically to a
position of $w_j$ by $L_t$; composing the two monotone matchings gives a common
subsequence of $w_i$ and $w_j$ of size $|S\cap L_t|$.
\end{proof}

The converse direction shows that the surrogate retains every strong
centre-based certificate needed by the sparsification argument. It relies on
the residual stopping condition in the construction of $Y_{c,b}$.

\begin{lemma}[Completeness]
\label{lem:complete}
Let $w_c \in W_A$ be a centre and suppose
$\lcs(\opt_{a,c}, w_b) \ge \lambda^2 D_{ab}/2$. Then
$I_c(a,b) = |\opt_{a,c}\cap Y_{c,b}| \ge \lambda^2 D_{ab}/4$.
\end{lemma}

\begin{proof}
The construction of $Y_{c,b}$ terminates when the residual satisfies
$\lcs(w_c\setminus Y_{c,b},w_b) < |w_c|\lambda^2/4 = d\lambda^2/4$. Split
$\opt_{a,c}$, which is a subsequence of $w_c$, into its part inside $Y_{c,b}$ and
its part outside. Any common subsequence of $\opt_{a,c}$ and $w_b$ splits
accordingly, so
\begin{align*}
\lcs(\opt_{a,c},w_b)
 \le  |\opt_{a,c}\cap Y_{c,b}| + \lcs(\opt_{a,c}\setminus Y_{c,b},w_b)
 <  |\opt_{a,c}\cap Y_{c,b}| + d\lambda^2/4 .
\end{align*}
Rearranging gives
$|\opt_{a,c}\cap Y_{c,b}| > \lambda^2 D_{ab}/2 - d\lambda^2/4 \ge \lambda^2 D_{ab}/4$,
the last inequality because $D_{ab}\ge d$, which holds for every mixed pair since one of its two windows lies in $W_A$ and has length exactly $d$ by Observation~\ref{obs:asymmetry}.
\end{proof}

\begin{corollary}[Substituting the surrogate]
\label{cor:surrogate}
Definition~\ref{def:surrogate} may replace the test of \cite[Algorithm 2]{rsss19}
without changing either the guarantee of \cite[Theorem 3.21]{rsss19} or the
conclusion of Lemma~\ref{lem:bipartite}.
\end{corollary}

\begin{proof}
The guarantee is Lemma~\ref{lem:sound}, which reproduces the bound
$\lcs(w_i,w_j)/\max\{|w_i|,|w_j|\}\ge\lambda^4/16$. For detection, the pair
produced in the proof of Lemma~\ref{lem:bipartite} satisfies
$\lcs(\opt_{a,c},w_b) \ge \lambda^2 D_{ab}/2$ for at least $k_A^{1-\gamma}$ centres
$w_c\in W_A$, and Lemma~\ref{lem:complete} converts each of these into
$I_c(a,b)\ge\lambda^2 D_{ab}/4$. This is precisely the certificate needed in the
repaired proof of Lemma~\ref{lem:bipartite}; no additional tie-breaking property
of the fixed longest common subsequences is used.
\end{proof}

\subsection{Sketching the ground set and batching by matrix multiplication}
\label{sec:mm}

The surrogate of Definition~\ref{def:surrogate} is an inner product of two $0/1$
vectors of dimension $|w_a| = d$, so a whole round is a single product of a
$k_A\times d$ matrix with a $d\times k_B$ matrix. We first reduce the inner
dimension.

Coordinate sampling reduces the surrogate's ambient dimension while preserving
the gap between its acceptance and rejection thresholds. Fresh independent
samples provide a simultaneous guarantee over all eligible centre--pair triples
in the run.

\begin{lemma}[Coordinate sketch]
\label{lem:sketch}
Fix a threshold and a dyadic $B$-layer with upper scale $d_q$, and call a mixed pair $(i,j)\in W_A\times W_B$ \emph{eligible} for the resulting run when its $B$-window lies in that layer, so that every eligible pair has $d_q/2<D_{ij}\le d_q$.  For a sufficiently large absolute constant $C$, put
\[
p_R:=\min\{1,2C\log n/(\lambda^2d_q)\}.
\]
If $p_R<1$, form $R\subseteq[d]$ by retaining each coordinate independently
with probability $p_R$; if $p_R=1$, take $R=[d]$. The sample is fresh for
this threshold, dyadic layer, and centre round, and is independent of the
input and the sampled centres. With probability at least $1-n^{-C}$,
simultaneously for every centre and every eligible pair $(i,j)$ in the run,
the sampled count
$|\opt_{i,a}\cap Y_{a,j}\cap R|$ exceeds
$3p_R\lambda^2D_{ij}/16$ whenever
$I_a(i,j)\ge\lambda^2D_{ij}/4$, and
$I_a(i,j)\ge\lambda^2D_{ij}/8$ whenever the sampled count exceeds this
entrywise threshold.
\end{lemma}

\begin{proof}
If $p_R=1$, the assertions hold deterministically, because the sampled count
equals the true count and $1/8<3/16<1/4$.  Otherwise
$p_R\lambda^2D_{ij}\ge C\log n$, since $D_{ij}>d_q/2$.  If
$I_a(i,j)\ge\lambda^2D_{ij}/4$, the displayed threshold is at most three
quarters of the expectation, so the lower-tail Chernoff bound accepts with
failure probability $\exp(-\Theta(\log n))$.  Conversely, if
$I_a(i,j)<\lambda^2D_{ij}/8$, the threshold is at least three halves of the
expectation, so the upper-tail Chernoff bound rejects with the same failure
probability.  (If the completeness premise is infeasible because
$\lambda^2D_{ij}/4>d$, its implication is vacuous.)  Conditional on the
sampled centres, there are at most $k^{2+\gamma}$ tested centre--pair triples
across all rounds.  A union bound completes the proof after choosing $C$
sufficiently large.  The same Chernoff bound and $d\le d_q$ give
$|R|=O(p_Rd+\log n)=\widetilde O(\lambda^{-2})$ with the required
probability.  Independence of $R$ from the input and the centres makes the
conditioning valid.
\end{proof}

The relaxed sampled threshold costs only a factor two: by
Lemma~\ref{lem:sound}, a pair marked through our sketch has
$\lcs(w_i,w_j)/D_{ij}\ge\lambda^4/32$.  The formal guarantee in
\cite[Theorem~3.21]{rsss19} is $\lambda^4/16$; the weaker constant here is
the price of separating the two Chernoff tails, and it is absorbed into the
$\Omega(\lambda^3)$ guarantee.  An accepted product entry therefore writes
the certified lower bound $D_{ij}\lambda^4/32$ into the estimate table.

After sketching, all surrogate tests from one centre round can be evaluated by
two rectangular matrix products. The resulting cost is essentially quadratic
in the number of windows, provided the sketched inner dimension lies below the
dual-exponent threshold.

\begin{lemma}[Cost of one sparsification round]
\label{lem:round}
Let $\alpha>0$ be the dual exponent of rectangular matrix multiplication and suppose the sketch dimension of the run satisfies $|R|\le k^{\alpha_0}$ for a fixed $\alpha_0<\alpha$. Then the batched products used by procedure \textnormal{\textsc{WindowStepOne}} in Algorithm~\ref{alg:window_step_one} execute one sparsification round in time $k^{2+o(1)}+\widetilde O(kd)$, so that $Q = k^{2+o(1)}$ in Eq.~\eqref{eq:T2}.
\end{lemma}

\begin{proof}
Write $w_c$ for the centre of the round. Restricting the indicator vectors of $\opt_{i,c}$ and $Y_{c,j}$ to $R$ costs $O(|\opt_{i,c}| + |Y_{c,j}|) = O(d)$ per window, hence $O(kd) = O(n)$ per round.
Let $U_A, U_B$ be the sketched indicator matrices of the sets $\opt_{i,c}$ for
$i\in W_A$ and $i\in W_B$ respectively, and let $V_A, V_B$ be those of the sets
$Y_{c,j}$. Both orientations are needed and both are obtained from two products:
$U_A V_B^{\top}$ supplies the sampled counts for
$(i,j)\in W_A\times W_B$, while $U_B V_A^{\top}$ supplies the reverse
certificates with $a\in W_B$, $b\in W_A$.  These raw counts are only tests;
for each product entry we compare against its own
$3p_R\lambda^2D_{ij}/16$ threshold, and an accepted test writes
$D_{ij}\lambda^4/32$ into the estimate table.  Since LCS is symmetric, a
reverse certificate is written into the unique $W_A\times W_B$ table entry
with the $A$-side index first; no second table is maintained.  Entries outside
the eligible mixed-pair mask of the current threshold/layer run are ignored.
A fresh independent $R$ is drawn for every centre round, as in Lemma~\ref{lem:sketch}.  Write $s:=|R|$. By Lemma~\ref{lem:sketch}, $s = \widetilde O(\lambda^{-2}) = n^{2a+o(1)}$, and $s\le k^{\alpha_0}$ by hypothesis; the caller in Lemma~\ref{lem:main} discharges this at its optimum. For such aspect ratios, Lemma~\ref{lem:rect_mm} computes the product of a
$k\times s$ by an $s\times k$ matrix in $k^{2+o(1)}$ time. To realize the products in the word-RAM model, choose a prime
$q>d$ (for $d\ge2$, Bertrand's postulate gives one below $2d$) and work over
$\mathbb F_q$. Every dot product of the $0/1$ vectors lies in
$[0,s]\subseteq[0,d]$, so reduction modulo $q$ causes no wraparound. Field
operations use $O(\log n)$-bit words, and their polylogarithmic overhead is
absorbed in the $o(1)$ exponent. The case $d=1$ is immediate.
\end{proof}

Fix once and for all a reciprocal power of two
$\epsilon_{\mathrm{win}}\le c_{\mathrm{str}}/8$, where
$c_{\mathrm{str}}$ is the absolute constant in
Lemma~\ref{lem:rsss_window_compatible}. For a dyadic density guess $\nu$, define
\begin{align*}
J_\nu&:=\max\{j\in\mathbb Z_{\ge0}:
 (1+\epsilon_{\mathrm{win}})^j\epsilon_{\mathrm{win}}\nu\le1\},\\
\Theta_\nu&:=
\{(1+\epsilon_{\mathrm{win}})^j\epsilon_{\mathrm{win}}\nu:
 0\le j\le J_\nu\}\cup\{1\}.
\end{align*}
Duplicates are removed. A threshold/layer combination is \emph{applicable}
to a pair when the pair belongs to that dyadic layer in the current RSSS
orientation. Thus $|\Theta_\nu|=O(\log n)$ and the definition is independent
of which marking computations are executed.

All arithmetic is exact in the word-RAM model. Each threshold is a rational
whose numerator and denominator have $O(\log n)$ bits; the same is true of
$D\theta^2$, $D\theta^4$, and their ceilings. For a fixed guess $\nu$, the
denominators are nested powers of two, so the estimate matrix may be scaled
by one $O(\log n)$-bit common denominator before the numerical block dynamic
program. Additions, maxima, and comparisons therefore use a constant number
of $O(\log n)$-bit words up to polylogarithmic overhead. This specifies the
nonnegative real entries without a real-RAM assumption.

\begin{remark}[Threshold parameterization]
\label{rem:threshold_parameterization}
Definition~\ref{def:surrogate}, Lemmas~\ref{lem:sound},
\ref{lem:complete}, \ref{lem:sketch}, and \ref{lem:round}, and
Corollary~\ref{cor:surrogate} are uniform in the threshold parameter.
For every $\theta\in(0,1]$, the same statements and proofs hold after
replacing $\lambda$ by $\theta$. Procedure
\textnormal{\textsc{WindowStepOne}} applies this parameterized form for
each $\theta\in\Theta_\nu$. Since $\theta\ge\epsilon_{\mathrm{win}}\nu$, the sampled
dimension is $\widetilde O(\theta^{-2})\le
\widetilde O(\nu^{-2})$; the fixed-crossover choice in Lemma~\ref{lem:main} ensures the required
rectangular matrix-multiplication aspect ratio.
\end{remark}

We stress that only $\alpha > 0$ is used, not any particular bound on $\omega$; the
resulting exponent is therefore independent of the value of $\omega$.

\begin{algorithm}[!ht]
\caption{Certified batched sparsification (window Step~1)}
\label{alg:window_step_one}
\begin{algorithmic}[1]
\Procedure{WindowStepOne}{$A,B,\nu,d,k,\gamma,\chi$}
\Comment{Lemma~\ref{lem:round}}
\State Construct $W_A,W_B$ and the dyadic $B$-layers for the current
scale and RSSS orientation $\chi$.
\State Initialize $M^{\mathrm{pre}}\gets 0$ and
$\mathcal S\gets\emptyset$.
\For{each $\theta\in\Theta_\nu$ and each dyadic $B$-layer $q$}
  \State Set $g\gets\lceil C_{\mathrm{ctr}}|W_A|^\gamma\log n\rceil$
  and sample $g$ centres independently and uniformly from $W_A$.
  \For{each sampled centre $w_c$}
    \State Compute every fixed alignment $\opt_{i,c}$ and every
    $Y_{c,j}$ using Lemma~\ref{lem:perpair}.
    \State Set $p_R\gets\min\{1,2C\log n/(\theta^2d_q)\}$ and draw a
    fresh $R\subseteq[d]$ by retaining each coordinate with probability $p_R$.
    \State Form $U_A,U_B,V_A,V_B$ on $R$, and compute
    $U_AV_B^\top$ and $U_BV_A^\top$.
    \For{each eligible $(i,j)\in W_A\times W_B$}
      \If{either oriented count is at least $\lceil3p_R\theta^2D_{ij}/16\rceil$}
        \State $M^{\mathrm{pre}}_{ij}\gets
        \max\{M^{\mathrm{pre}}_{ij},D_{ij}\theta^4/32\}$.
      \EndIf
    \EndFor
  \EndFor
  \State Add the threshold, layer, actual lengths, and mark bits of this
  run to $\mathcal S$.
\EndFor
\State \Return $(M^{\mathrm{pre}},\mathcal S)$.
\EndProcedure
\end{algorithmic}
\end{algorithm}

\begin{remark}
\label{rem:nostack}
One cannot batch several rounds into a single product. Concatenating the sketches
of $g$ rounds computes $\sum_a I_a(i,j)$, and a pair may clear the threshold in the
sum without clearing it in any single round; Lemma~\ref{lem:sound} then fails and
the estimate matrix violates the requirement $M(i,j)\le\lcs(w_i,w_j)$ of
\cite[Algorithm 6]{rsss19}, so the algorithm may output more than $\lcs(A,B)$.
Grouping $g$ rounds and lowering the threshold degrades the guarantee by a factor
$g$, so $g = O(1)$.
\end{remark}

\subsection{The small-\texorpdfstring{$\lambda$}{lambda} regime}
\label{sec:small-lambda}

When $\lambda$ is small the window machinery becomes expensive and one switches to
a direct approximation. \cite{rsss19} samples $A$ at rate $\lambda^3$ and extends,
which costs $\widetilde O(n^2\lambda^7)$, that is $r = 7$ in Eq.~\eqref{eq:T0}. The
tradeoff of Bringmann, Cohen-Addad and Das is better.

\begin{corollary}[Small-density tradeoff]
\label{cor:r}
For $a \le 2/15$, an approximation factor $\lambda^{-3} = n^{3a}$ is achievable in
time $n^{2-15a/2}$, that is, $r = 15/2$ in Eq.~\eqref{eq:T0}. At the optimum
$a = 4/189$, so the restriction is not binding.
\end{corollary}

\begin{proof}
Set $2\varepsilon/5 = 3a$ in Lemma~\ref{lem:bcd}.
\end{proof}

The following endpoint supplies the direct branch when the LCS density is below
a fixed polynomial threshold. It specializes the BCAD approximation to
$n^{8/5}$ time while retaining a $\lambda^{-3}$ approximation factor and a
one-sided output.

\begin{theorem}[BCAD at a fixed small-density crossover]
\label{thm:small_density_eight_fifths}
For every fixed $a>4/75$ and every fixed failure exponent $C>0$, there is a
randomized algorithm that, given length-$n$ strings with
$\lambda=\lcs(A,B)/n\le n^{-a}$, returns a number $Z$ satisfying
\[
\lambda^3\lcs(A,B)\le Z\le\lcs(A,B)
\]
with probability at least $1-n^{-C}$, in time $O_{a,C}(n^{8/5})$.
\end{theorem}
\begin{proof}
Run the algorithm of Lemma~\ref{lem:bcd} with $\varepsilon=2/5$. Its runtime
is $O(n^{2-2/5})=O(n^{8/5})$, and for some absolute constants $K,q>0$ its
one-sided output is at least
\[
\frac{\lcs(A,B)}{K n^{4/25}\log^q n}
\]
with high probability. Put $\delta:=3a-4/25>0$. For all
$n\ge n_0(a)$, one has $K\log^q n\le n^\delta$, and therefore the output is
at least
\[
n^{-4/25-\delta}\lcs(A,B)
=n^{-3a}\lcs(A,B)
\ge\lambda^3\lcs(A,B).
\]
The output is one-sided by Lemma~\ref{lem:bcd}. A constant number depending on
$C$ of independent repetitions, followed by taking the maximum, gives failure
probability at most $n^{-C}$. Inputs below $n_0(a)$ are solved exactly; because
there are only finitely many such lengths, their quadratic cost is absorbed in
the $a$-dependent constant multiplying $n^{8/5}$.
\end{proof}

Both the window algorithm and Lemma~\ref{lem:bcd} output a number that is at most
$\lcs(A,B)$, so the overall algorithm may run both and return the larger of the two;
the guarantee holds as soon as one of the two is good. Because the factor of Lemma~\ref{lem:bcd} is $\widetilde O(n^{3a})$ rather than $n^{3a}$, the window branch is run for every guess $\nu\ge n^{-a}\log^{-c}n$ and Lemma~\ref{lem:bcd} is used only below that. This raises the window exponent by $O(\log\log n/\log n)$, hence the running time by a polylogarithmic factor, and it leaves $\lambda^3\le n^{-3a}\log^{-3c}n$ wherever Lemma~\ref{lem:bcd} is applied, which absorbs the polylogarithmic factor in its approximation guarantee once the absolute constant $c$ is large enough.

\subsection{The ANSS endpoint}
\label{sec:anss_endpoint}

The next theorem packages the ANSS LIS estimator into the form consumed by the
low-density branch. The reduction makes the approximation loss explicit and
obtains high-probability correctness together with a worst-case time cap.

\begin{theorem}[Quantitative ANSS endpoint for LCS]
\label{thm:anss_lcs_endpoint}
The following holds with the absolute constants
$c_{\mathrm A}=1/2$ and $K_{\mathrm A}=2$. For every fixed $C>0$ there is a
randomized algorithm
which, on two strings $A,B\in\Sigma^n$, returns an integer $Z_{\mathrm A}$
such that, for all sufficiently large $n$,
\[
 c_{\mathrm A}
 \frac{\lcs(A,B)}{n^{K_{\mathrm A}/\sqrt{\log\log n}}}
 \le Z_{\mathrm A}\le \lcs(A,B)
\]
with probability at least $1-n^{-C}$.  Its worst-case running time is
$O_C(n^{4/3}\log n)$ in the random-access word-RAM model.  In particular,
for every fixed $a>0$ the lower bound is at least
$c_{\mathrm A}n^{-3a}\lcs(A,B)$ for all sufficiently large $n$.
\end{theorem}

\begin{algorithm}[!ht]
\caption{One execution of the ANSS small-$\lambda$ endpoint}
\label{alg:small_lambda_estimate}
\begin{algorithmic}[1]
\Procedure{SmallLambdaEstimate}{$A,B\in\Sigma^n$}
  \Comment{Single-run core of Theorem~\ref{thm:anss_lcs_endpoint}}
  \If{$n<2$}
    \State \Return $\lcs(A,B)$.
  \EndIf
  \State Compute the increasing occurrence lists
  $P_\sigma\gets\{j\in[n]:B_j=\sigma\}$ for all symbols $\sigma$.
  \State Set $s_0\gets0$ and
  $s_i\gets s_{i-1}+|P_{A_i}|$ for $i\in[n]$; put $N\gets s_n$.
  \State Regard $Z\in[n]^N$ as the virtual concatenation of
  $P_{A_1},\ldots,P_{A_n}$, each list written in decreasing order.
  \State Answer $Z[p]$ by finding the unique $i$ with
  $s_{i-1}<p\le s_i$ and reading entry $p-s_{i-1}$ of $P_{A_i}$ in reverse.
  \State Set $\ell\gets\lceil\log n\rceil$.
  \If{$N\le2n\ell$}
    \State Materialize $Z$ and \Return its strict LIS length.
  \EndIf
  \State Set $L_0\gets N/(2n)$, $\mu\gets1/(2n\ell)$, and
  $\epsilon\gets1/\log\log(1/\mu)$.
  \State Run one complete fresh execution of the estimator of
  Lemma~\ref{lem:anss} on virtual $Z$; let $\widehat L$ be its output.
  \State \Return $\max\{\lfloor L_0\rfloor,\widehat L\}$.
\EndProcedure
\end{algorithmic}
\end{algorithm}

\begin{algorithm}[!ht]
\caption{Capped-median quantitative small-$\lambda$ endpoint}
\label{alg:quantitative_small_lambda_estimate}
\begin{algorithmic}[1]
\Procedure{QuantitativeSmallLambdaEstimate}{$A,B,C$}
  \Comment{Theorem~\ref{thm:anss_lcs_endpoint}}
  \State Set $B_{\mathrm{cap}}\gets\lceil n^{4/3}\rceil$ and
  $q\gets2\lceil18(C+1)\log n\rceil+1$.
  \For{$r=1,\ldots,q$}
    \State Run an independent call
    $Y_r\gets\Call{SmallLambdaEstimate}{A,B}$ for at most
    $B_{\mathrm{cap}}$ word operations.
    \Comment{Algorithm~\ref{alg:small_lambda_estimate}}
    \If{the call reaches the cap}
      \State Abort it and set $Y_r\gets0$.
    \EndIf
  \EndFor
  \State \Return $\operatorname{median}(Y_1,\ldots,Y_q)$.
\EndProcedure
\end{algorithmic}
\end{algorithm}

\begin{proof}
Write $L:=\lcs(A,B)$.  For $i\in[n]$, let $Z_i$ be the list of positions
$j\in[n]$ for which $B_j=A_i$, written in decreasing order, and concatenate
these lists to form $Z:=Z_1\circ\cdots\circ Z_n$.  Put
$M:=|Z|$.  We first verify the two elementary properties of this reduction
that will be used below.

By the convention in Section~\ref{sec:basic_notation}, an increasing
subsequence is strictly increasing. Such a subsequence of $Z$ uses at most
one entry from each decreasing block $Z_i$. Thus its block indices and its values give increasing positions
in $A$ and $B$ carrying the same symbols, and hence a common subsequence.
Conversely, the matched positions of any common subsequence select one entry
from each of the corresponding blocks and form a strictly increasing subsequence of
$Z$.  Therefore
\begin{equation}
 \lis(Z)=L. \label{eq:anss_flatten}
\end{equation}
For a symbol $\sigma$, let $f_A(\sigma)$ and $f_B(\sigma)$ be its two
frequencies and put
$a_\sigma:=\min\{f_A(\sigma),f_B(\sigma)\}$ and
$b_\sigma:=\max\{f_A(\sigma),f_B(\sigma)\}$.  Then
\begin{equation}
 M=\sum_\sigma a_\sigma b_\sigma
 \le \max_\sigma a_\sigma\sum_\sigma b_\sigma
 \le 2n\max_\sigma a_\sigma
 \le 2nL. \label{eq:anss_match_pairs}
\end{equation}
The last inequality follows by taking the common subsequence consisting only
of a symbol attaining $\max_\sigma a_\sigma$.

All occurrence lists of symbols in $B$ and their prefix sums across the
blocks $Z_i$ can be constructed in $O(n\log n)$ time.  Given a position in
$Z$, predecessor search among the prefix sums locates its block, after which
the requested value is read from the appropriate occurrence list in reverse
order.  This gives random access to the virtual sequence $Z$ in $O(\log n)$
time without materializing it.

Let $\ell:=\lceil\log n\rceil$.  If $M\le2n\ell$, we materialize $Z$ and
compute its strict LIS exactly in $O(M\log M)=O(n\log^2n)$ time.  Assume
henceforth that $M>2n\ell$, and set
\[
 N:=M,\qquad \mu:=\frac1{2n\ell},\qquad L_0:=\frac{M}{2n}.
\]
Then $1/N<\mu<1$, and by~\eqref{eq:anss_match_pairs},
\begin{equation}
 \lis(Z)=L\ge L_0=\mu N\ell>\mu N. \label{eq:anss_promise}
\end{equation}

We apply Lemma~\ref{lem:anss} with $N=M$, the value of $\mu$ above, and its
prescribed parameter schedule
\[
 \epsilon=\frac1{\log\log(1/\mu)}.
\]
One complete execution has a joint success event on which its output
$\widehat L$ satisfies
\begin{equation}
 \frac{L}{\alpha}-\mu N\le\widehat L\le L, \qquad
 \alpha=(1/\mu)^{\sqrt\epsilon}
 (\log(1/\mu))^{2^{O(\log^2(1/\epsilon)/\sqrt\epsilon)}}.
 \label{eq:anss_parameterized}
\end{equation}
The probability of this event is $1-o(1)$, and the total expected running
time of the complete recursive call, including its precision-tree
preprocessing, is $(1/\mu)N^{o(1)}$; the latter is the expected-time
calculation in Appendix~A of~\cite{anss22}.  Both inequalities in
\eqref{eq:anss_parameterized} belong to the same success event.  We do not
assume that the upper inequality holds for every outcome.

We next make the hidden rate in~\eqref{eq:anss_parameterized} explicit at
this particular schedule.  Put
\[
 u:=\log(1/\mu),\qquad x:=\log u.
\]
Thus $\epsilon=1/x$.  If
$T:=2^{O(\sqrt{x}(\log x)^2)}$, then
\[
 \log\alpha=\frac{u}{\sqrt{x}}+xT.
\]
The displayed approximation formula is Theorem~4.3 of~\cite{anss22}. The
fact needed here that its hidden constants are independent of $\epsilon$
follows from Claim~A.1 and the recursion analysis in Appendix~A of the full
version. Since
$u=e^x$ and $T=\exp(O(\sqrt{x}(\log x)^2))$, we have
$xT=o(u/\sqrt{x})$.  Here
\[
 u=\log(2n\ell)=\log n+O(\log\log n),\qquad
 x=\log\log n+o(1).
\]
Consequently, for all sufficiently large $n$,
\begin{equation}
 \frac{\log\alpha}{\log n}
 =\frac{1+o(1)}{\sqrt{\log\log n}},
 \qquad
 \alpha\le n^{2/\sqrt{\log\log n}}.
 \label{eq:anss_rate}
\end{equation}
At the same schedule, the expected time of
one execution, including the $O(\log n)$ overhead for each virtual access,
is
\[
 (1/\mu)N^{o(1)}\log n
 \le 2n\ell\,(n^2)^{o(1)}\log n=n^{1+o(1)}.
\]

It remains to convert expectation and the one-run success probability into
the claimed worst-case and failure bounds.  Cap a fresh execution after
$B:=\lceil n^{4/3}\rceil$ operations and make a timed-out execution output
$0$.  Markov's inequality bounds its timeout probability by
$n^{-1/3+o(1)}$.  Since the joint success probability in
\eqref{eq:anss_parameterized} tends to one, for all sufficiently large $n$
a capped execution completes and is jointly successful with probability at
least $2/3$.

Take an odd number
\[
 q:=2\lceil 18(C+1)\log n\rceil+1
\]
of independent capped executions, using fresh randomness throughout, and
let $\widetilde L$ be the median of their outputs.  Hoeffding's inequality
shows that the probability that at most half are complete and successful is
at most $\exp(-q/18)\le n^{-(C+1)}$.  If a strict majority are successful,
then a strict majority of their values lie in the interval
\[
 [L/\alpha-L_0/\ell,\,L],
\]
and therefore the median lies in the same interval.  Finally return
\[
 Z_{\mathrm A}:=\max\{\lfloor L_0\rfloor,\widetilde L\}.
\]
By~\eqref{eq:anss_match_pairs}, $\lfloor L_0\rfloor\le L$, so on the
majority-good event $Z_{\mathrm A}\le L$.  If
$L\ge2\alpha L_0/\ell$, then
$\widetilde L\ge L/(2\alpha)$.  Otherwise
$L_0>L\ell/(2\alpha)$; because the nonexact branch has $L_0>\ell$, for
large $n$ we have
\[
 \lfloor L_0\rfloor\ge L_0/2>\frac{L\ell}{4\alpha}
 \ge\frac{L}{2\alpha}.
\]
Thus in both cases $Z_{\mathrm A}\ge L/(2\alpha)$. Hence
$c_{\mathrm A}=1/2$ and $K_{\mathrm A}=2$ give the displayed approximation
guarantee. The union of the finitely
many exceptional input lengths is handled by exact computation.

Every capped execution takes at most $B$ operations.  Hence the amplified
worst-case time is
\[
 qB=O_C(n^{4/3}\log n).
\]
Preprocessing and the exact branch are smaller.  Finally, for each fixed
$a>0$, $K_{\mathrm A}/\sqrt{\log\log n}\le3a$ for all sufficiently large
$n$, which proves the last assertion.
\end{proof}

\subsection{The explicit nonrecursive bound}
\label{sec:optimization}

Write every cost as an exponent of $n$. By Lemma~\ref{lem:perpair},
$P = \widetilde O(d^2 + w_{\max})$, and at the optimum below $d^2$ dominates, so
$P = n^{2(1-x)}$. By Lemma~\ref{lem:round}, $Q = k^{2+o(1)}$. By
Eq.~\eqref{eq:T3-final}, $q' = 9/2$. With $k = n^x$ and $\lambda = n^{-a}$ the four
exponents of Eq.~\eqref{eq:T1}--\eqref{eq:T0} are
\begin{align}
t_1 &= x(1+\gamma) + 2(1-x) = 2 - x + x\gamma, \label{eq:t1}\\
t_2 &= x\gamma + 2x = x(2+\gamma), \label{eq:t2}\\
t_3 &= 2(1-x) + x(2-\gamma/2) + \tfrac92 a = 2 - x\gamma/2 + \tfrac92 a,\\
t_0 &= 2 - \tfrac{15}{2}a. \label{eq:t0}
\end{align}
The running time is $\widetilde O(n^{E})$ with $E = \max\{t_1,t_2,t_3,t_0\}$, and
we minimize over $x\in(0,1)$, $\gamma\in(0,1)$ and the crossover $a$.

The next lemma records the exponent obtained by optimizing the window framework
without the recursive density reduction. It provides a concrete ANSS-free
benchmark while checking that all parameter side conditions hold simultaneously.

\begin{lemma}[Explicit nonrecursive bound]
\label{lem:main}
\begin{sloppypar}
There is a randomized algorithm that, given strings $A,B$ of length $n$ with
$\|\lcs(A,B)\| = \lambda$, computes an $\Omega(\lambda^3)$-approximation of
$\lcs(A,B)$ in time $\widetilde O(n^{116/63})$ with high probability.
\end{sloppypar}
\end{lemma}

\begin{proof}
From Eq.~\eqref{eq:t1} and Eq.~\eqref{eq:t2}, $t_1 \ge t_2$ if and only if $2-x\ge 2x$,
i.e.\ $x\le 2/3$, independently of $\gamma$ and $a$. Both $t_1$ and $t_3$ are
decreasing in $x$, so the optimum takes
\begin{align}
\label{eq:xstar}
x = 2/3 .
\end{align}
Setting $t_1 = t_3$ gives $-x + \tfrac32 x\gamma = \tfrac92 a$, that is
$\gamma = \tfrac23 + 3a/x$, so with Eq.~\eqref{eq:xstar}
\begin{align*}
\gamma  =  \tfrac23 + \tfrac92 a ,
\end{align*}
and then $t_1 = 2 - x/3 + 3a = \tfrac{16}{9} + 3a$. Setting this equal to $t_0$
from Eq.~\eqref{eq:t0} gives $\tfrac{21}{2}a = \tfrac29$, hence
\begin{align*}
a = \frac{4}{189},\qquad \gamma = \frac{16}{21},\qquad x = \frac23 ,
\end{align*}
and the common value is
\begin{align*}
E  =  2 - \frac{15}{2}\cdot\frac{4}{189}  =  2 - \frac{10}{63}  =  \frac{116}{63}.
\end{align*}
All four exponents agree at this point:
$t_1 = 2 - \tfrac23\cdot\tfrac{5}{21} = 2 - \tfrac{10}{63}$,
$t_2 = \tfrac23\cdot\tfrac{58}{21} = \tfrac{116}{63}$,
$t_3 = 2 - \tfrac{16}{63} + \tfrac{6}{63} = 2 - \tfrac{10}{63}$, and
$t_0 = 2 - \tfrac{30}{189} = 2 - \tfrac{10}{63}$.

It remains to check the side conditions. We have $d = n^{1/3}$, $k = n^{2/3}$ and
$w_{\max} = n^{1/3 + 4/189}$, so $d^2 = n^{2/3}$ dominates $w_{\max}$ in
Lemma~\ref{lem:perpair}, and $nd = n^{4/3}$ dominates $n/\lambda = n^{1+4/189}$ in
its proof. Also $\gamma = 16/21 \in (0,1)$, as required by
\cite[Theorem 3.21]{rsss19}. For Step~2, Lemma~\ref{lem:bipartite} gives
$\eta = \gamma - a/x = 46/63$, so the nearby radius is
$W = k^{\eta/2}d = n^{109/189}$.  In particular
$w_{\max}=n^{67/189}\le W\le n$, which makes
Observation~\ref{obs:nearby-counts} applicable, and the
sampling rate is
$p = \widetilde O(1/(\epsilon\lambda^4 k^{\eta/2})) = \widetilde O(n^{16/189-46/189})
\le 1$, as a probability must be. Finally $s = \widetilde O(n^{2a}) =
\widetilde O(n^{8/189})$ and $k^{\alpha} = n^{2\alpha/3}$, so the hypothesis
$s \le k^\alpha$ of Lemma~\ref{lem:round} reads $\alpha \ge 4/63 = 0.0635\ldots$,
which holds with room to spare for the known $\alpha > 0.321$ \cite{wxxz24}.

For the window branch, run the construction on the four transformed pairs
$(A,B)$, $(B,A)$, $(\operatorname{rev}A,\operatorname{rev}B)$ and
$(\operatorname{rev}B,\operatorname{rev}A)$ and take the largest numerical
estimate.  Lemma~4.7 of \cite{rsss19} supplies the required compatible window
path for at least one orientation.  Swap and reversal preserve LCS, and every
run is one-sided, so this wrapper costs only a factor four and preserves
soundness.

For $\lambda \le n^{-4/189}$ the algorithm of Corollary~\ref{cor:r} is used and
costs at most $n^{116/63}$; for $\lambda \ge n^{-4/189}$ the four-orientation
window branch costs at most $n^{116/63}$. Since $\lambda$ is unknown the
algorithm runs both and returns the larger of the numbers, which is again at
most $\lcs(A,B)$; the approximation guarantee is $\Omega(\lambda^3)$ in
either case.
\end{proof}

\subsection{Adaptive centre domination}
\label{sec:adaptive_domination}

The following graph records when two windows aligned to the same anchor share
many anchor positions. Its edges allow one selected centre to expose and
eliminate a large class of candidates at once.

\begin{definition}[Anchored intersection graph]
\label{def:anchor_graph}
Fix $\theta\in\Theta_\nu$ and an anchor $w_a\in W_B$, and write $D_a:=\max\{|w_a|,d\}$; since every window of $W_A$ has length exactly $d$ by Observation~\ref{obs:asymmetry}, this is the common value of $D_{ba}$ over all $w_b\in W_A$, and $d\le D_a$.
Put $N_\theta(a):=\{b:w_b\in W_A,\ \lcs(w_b,w_a)\ge\theta D_a\}$ and, for
$b\in N_\theta(a)$, let $S_b:=\opt_{b,a}\subseteq[D_a]$ be the positions of
$w_a$ used by the canonical alignment of the pair, so that
$|S_b|\ge\theta D_a$. The graph $H_{a,\theta}$ has vertex set $N_\theta(a)$,
with $b$ and $c$ adjacent whenever $|S_b\cap S_c|\ge\theta^2D_a/2$.
\end{definition}

The overlap threshold prevents a large family of anchor alignments from being
pairwise nonadjacent. This bound supplies the finite number of domination stages
used later.

\begin{lemma}[Bounded independence]
\label{lem:anchor_independence}
$\alpha(H_{a,\theta})<\lceil2/\theta\rceil$.
\end{lemma}
\begin{proof}
Let $T\subseteq N_\theta(a)$ with $t:=|T|\ge\lceil2/\theta\rceil$ and set
$m(x):=|\{b\in T:x\in S_b\}|$, so that
$Z:=\sum_{x}m(x)=\sum_{b\in T}|S_b|\ge t\theta D_a\ge2D_a$. By
Cauchy--Schwarz,
\begin{align*}
\sum_{b\ne c}|S_b\cap S_c|
=\sum_xm(x)^2-\sum_xm(x)
\ge\frac{Z^2}{D_a}-Z
\ge t\theta D_a(t\theta-1),
\end{align*}
the last step because $z\mapsto z(z/D_a-1)$ increases for $z\ge D_a/2$.
Averaging over fewer than $t^2$ ordered pairs, some $b\ne c$ satisfies
$|S_b\cap S_c|\ge\theta D_a(\theta-1/t)\ge\theta^2D_a/2$ and is therefore an
edge. Hence no $\lceil2/\theta\rceil$ vertices are pairwise nonadjacent.
\end{proof}

Adjacency certifies that the two corresponding alignments share many positions
of the anchor. Transporting those shared positions through one centre creates a
large mark for the other endpoint.

\begin{lemma}[Edges are visible to one centre]
\label{lem:edge_marks}
If $b$ and $c$ are adjacent in $H_{a,\theta}$ then
$I_c(a,b)\ge\theta^2D_a/4$.
\end{lemma}
\begin{proof}
Each $x\in S_b\cap S_c$ is used by both canonical alignments; mapping it
through the alignment of $\{a,c\}$ into $w_c$ and through that of $\{b,a\}$
into $w_b$ matches equal characters, monotonically in $x$. The resulting
common subsequence of $\opt_{a,c}$ and $w_b$ has length $|S_b\cap S_c|$, so
$\lcs(\opt_{a,c},w_b)\ge\theta^2D_a/2$, and Lemma~\ref{lem:complete} in the
form of Remark~\ref{rem:threshold_parameterization} gives the claim.
\end{proof}

The next abstract lemma turns a bounded-independence hypothesis into a
sampling-based domination procedure. Independent pools shrink the residual graph
to a prescribed size with only logarithmic sampling overhead.

\begin{lemma}[Random pool domination]
\label{lem:pool_domination}
Let $G$ be a graph whose vertices lie in a universe of size $k$, with
$\alpha(G)<r$, and let $1\le q\le k$. Draw $r$ independent pools of
$\lceil C_h(k/q)\log n\rceil$ uniform elements of the universe and run $r$
stages, at stage $t$ selecting a vertex of $G$ that lies in pool $t$ and is
still residual, then deleting it together with all of its neighbours. With
probability $1-n^{-\Omega(C_h)}$ at most $q$ vertices remain residual.
\end{lemma}
\begin{proof}
If more than $q$ vertices remain at the end, then more than $q$ are residual
before every stage, since residual sets only shrink. Pool $t$ is drawn
independently of the first $t-1$ pools, hence of the residual set entering
stage $t$, and misses a set of more than $q$ vertices with probability at
most $(1-q/k)^{C_h(k/q)\log n}\le n^{-C_h}$. On the complementary event a
vertex is selected at every stage. Any neighbour of an earlier selection has
been deleted, so the $r$ selections are pairwise nonadjacent, contradicting
$\alpha(G)<r$.
\end{proof}

The following filter tests all candidate windows against one preprocessed centre
without recomputing every alignment. Random coordinate sampling separates the
two relevant overlap thresholds within the stated aggregate cost.

\begin{lemma}[Sparse one-centre filter]
\label{lem:sparse_filter}
Fix $\theta$, an anchor $w_a\in W_B$ and a centre $w_c\in W_A$. Assume
that the threshold-dependent inverted lists $\mathcal L^{(\theta)}_{c,x}$
have already been built as part of the centre preprocessing. There is a
procedure running in time
$\widetilde O(k\theta^{-2}+d)$ that, with probability $1-n^{-C}$, accepts
every $b$ with $I_c(a,b)\ge\theta^2D_a/4$ and accepts no $b$ with
$I_c(a,b)<\theta^2D_a/8$.
\end{lemma}
\begin{proof}
Use the prebuilt lists
$\mathcal L^{(\theta)}_{c,x}:=\{b:x\in Y^{(\theta)}_{c,b}\}$, which satisfy
$\sum_x|\mathcal L^{(\theta)}_{c,x}|=\sum_b|Y^{(\theta)}_{c,b}|\le kd$. Retain each
$x\in\opt_{a,c}$ independently with probability
$p:=\min\{1,2C_R\log n/(\theta^2D_a)\}$, initialize every
counter to zero, traverse $\mathcal L_{c,x}$ for every retained $x$,
incrementing the counter for each $b$ met, and accept $b$
when its counter reaches $\tfrac3{16}p\theta^2D_a$. If $p=1$ the counter equals
$I_c(a,b)$ and $\tfrac18<\tfrac3{16}<\tfrac14$ gives both implications
deterministically. If $p<1$ then $p\theta^2D_a=2C_R\log n$; a pair with
$I_c(a,b)\ge\theta^2D_a/4$ has mean at least $\tfrac14p\theta^2D_a$, of
which the cutoff is three quarters, and a pair with
$I_c(a,b)<\theta^2D_a/8$ has mean below $\tfrac18p\theta^2D_a$, of which the
cutoff is three halves, so Chernoff bounds and a union bound over $b$ give
failure probability $n^{-\Omega(C_R)}$. The number of retained coordinates is a sum of independent indicators of mean $p|\opt_{a,c}|\le pd\le2C_R\log n/\theta^2$, by $d\le D_a$. Consequently a Chernoff upper-tail bound shows that it exceeds the deterministic cutoff $4C_R\theta^{-2}\log n$ with probability $n^{-\Omega(C_R)}$. Each retained coordinate costs $|\mathcal L_{c,x}|\le k$, the scan of $\opt_{a,c}$ costs $O(d)$ and reading the counters costs $O(k)$, which gives the stated bound; a run exceeding this cutoff is aborted.
\end{proof}

The next lemma combines bounded independence, random-pool domination, and the
sparse filter into the first adaptive stage. It controls both the preprocessing
cost and the number of high pairs left unmarked across every threshold and
dyadic layer.

\begin{lemma}[Adaptive Step~1]
\label{lem:adaptive_step_one}
Procedure \textnormal{\textsc{AdaptiveStepOne}} in Algorithm~\ref{alg:adaptive_step_one} preprocesses $\widetilde O(k^\gamma)$ centres of $W_A$ and runs in time
\begin{align*}
\widetilde O(k^{1+\gamma}d^2+k^\gamma n/\nu+k^2\nu^{-3}),
\end{align*}
such that, with high probability, simultaneously
$M^{\mathrm{pre}}_{ij}\le\lcs(w_i,w_j)$ for every entry and, for every
$\theta\in\Theta_\nu$, every dyadic layer and
every anchor $w_a$, all but $\widetilde O(k^{1-\gamma}/\theta)$ members of
$N_\theta(a)$ are marked at $\theta$; equivalently, there are absolute constants $C_U>0$ and $q_U:=4$ such that
the union, over the full threshold net and all dyadic layers, of mixed pairs
that are high and unmarked at their assigned largest threshold has size at
most
\[
C_U k^{2-\gamma}\nu^{-1}\log^{q_U}n.
\]
Procedure \textnormal{\textsc{AdaptiveStepOne}} returns this information in
the complete orientation-specific instance descriptor defined in
Algorithm~\ref{alg:adaptive_step_one}.
\end{lemma}
\begin{proof}
Fix $\theta$ and a layer, write $k_A:=|W_A|$, and put
$q_\theta:=\lceil C_qk_A^{1-\gamma}/\theta\rceil$.
If $q_\theta<k_A$, put $r_\theta:=\lceil2/\theta\rceil$ and
$g_\theta:=\lceil C_h(k_A/q_\theta)\log n\rceil$; then
$g_\theta=\widetilde O(k_A^\gamma\theta)$.  Otherwise the assertion is vacuous. Draw $r_\theta$ pools of $g_\theta$ uniform windows of $W_A$; since $1/\theta=O(k_A^\gamma)$ in this regime, the total number of
centre occurrences is
$r_\theta g_\theta=\widetilde O(k_A^\gamma)=\widetilde O(k^\gamma)$. Preprocessing one centre costs $\widetilde O(kd^2+n/\nu)$ by Lemma~\ref{lem:perpair}: the alignments $\opt_{a,c}$ against the $\widetilde O(k)$ windows of $W_B$ cost $\widetilde O(kd^2+\sum_a|w_a|)$ and $\sum_a|w_a|=\widetilde O(nw_{\max}/d)=\widetilde O(n/\nu)$ by the layer counts of Observation~\ref{obs:nearby-counts}, while the threshold-dependent packings $Y^{(\theta)}_{c,b}$ against $W_A$ cost $\widetilde O(kd^2)$. The alignments do not depend on $\theta$, whereas $Y_{c,b}$ and the inverted lists of Lemma~\ref{lem:sparse_filter} are rebuilt for each of the polylogarithmically many thresholds. For each
anchor run $r_\theta$ stages, at stage $t$ selecting the first centre of
pool $t$ that lies in $N_\theta(a)$ and is unmarked, marking it by its exact
precomputed value, and running Lemma~\ref{lem:sparse_filter} for the pair;
every accepted $b$ receives $\theta^4D_a/32$. Acceptance gives
$I_c(a,b)\ge\theta^2D_a/8$, so the written value is at most
$\lcs(w_a,w_b)$ by Lemma~\ref{lem:sound}, and one-sidedness holds. By
Lemma~\ref{lem:edge_marks} the filter accepts every $H_{a,\theta}$-neighbour
of the selected centre, so a stage realizes the deletion step of
Lemma~\ref{lem:pool_domination} for $G=H_{a,\theta}$, whose independence
number is below $r_\theta$ by Lemma~\ref{lem:anchor_independence}. The filter also accepts windows that are not neighbours, which only deletes more and leaves the conclusion intact, and its pools are drawn independently of all filter coins, so the pool used at a stage is independent of the residual set entering it. That lemma leaves at most $q_\theta$ members of $N_\theta(a)$ unmarked. Since the number of anchor-stage pairs is $\widetilde O(k/\theta)=n^{O(1)}$, choosing $C_h$ and $C_R$ large in terms of the desired failure exponent makes the union bound over all filters, stages, anchors, thresholds, layers, density guesses and orientations lose only $n^{-\Omega(1)}$. Summing over the anchors, the $O(\log n)$ thresholds and the $O(\log n)$
dyadic layers, and absorbing ceilings gives, for a sufficiently large absolute
$C_U$, the explicit bound
$C_Uk^{2-\gamma}\nu^{-1}\log^{q_U}n$ with $q_U=4$; in particular it bounds
the union used by the staircase.  Moreover, a pair marked at $\theta$ receives $\theta^4D/32$; this numerical lower bound is the only marked-pair fact used later. Equivalently, the loops may be reordered stage-major: pool $t$ can be
preprocessed, used by every anchor, and then discarded. This changes only the
memory organization, not the procedure or its running time. For the running time, each anchor invokes $O(1/\theta)$ filters of cost $\widetilde O(k\theta^{-2}+d)$ and
scans $\widetilde O(k^\gamma)$ pool entries, so one threshold costs
$\widetilde O(k^2\theta^{-3}+kd\theta^{-1}+k^{1+\gamma})$; summing the
geometric net and the $O(\log n)$ layers gives the bound.  The same
union bound includes every sparse-filter soundness event.  Conditioned on
their intersection, each accepted counter writes at most the corresponding
window LCS by Lemma~\ref{lem:sparse_filter}; hence the whole matrix
$M^{\mathrm{pre}}$ is one-sided on this joint event.
\end{proof}

\begin{figure}[!ht]
\centering
\includegraphics[width=\linewidth]{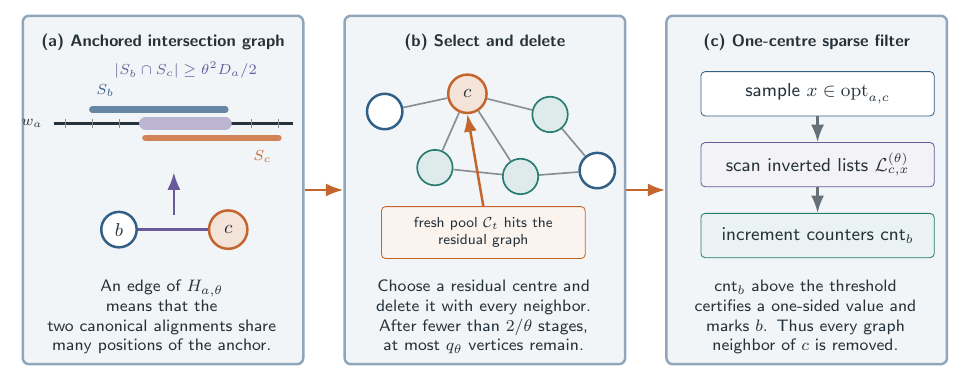}
\caption{Adaptive centre domination in Lemma~\ref{lem:adaptive_step_one}. Anchor-alignment overlap defines the graph $H_{a,\theta}$; each fresh pool selects a residual centre and eliminates its neighborhood through the sparse one-centre filter.}
\label{fig:adaptive_centre_domination}
\end{figure}

\begin{algorithm}[!ht]
\caption{Adaptive centre domination (window Step~1)}
\label{alg:adaptive_step_one}
\begin{algorithmic}[1]
\Procedure{AdaptiveStepOne}{$A,B,\nu,d,k,\gamma,\chi,\epsilon_{\mathrm{win}}$}
\Comment{Lemma~\ref{lem:adaptive_step_one}}
\State Take the orientation $\chi$ of the common padded pair supplied by
the outer wrapper (Algorithm~\ref{alg:approximate_lcs} or its fixed
three-scale specialization), obtaining $A_\chi,B_\chi$, and construct
$W_A,W_B$ with window parameter
$\epsilon_0=\epsilon_{\mathrm{win}}\nu$ and the dyadic $B$-layers. Do not
re-pad after orienting. Rows index $W_A$ and columns index $W_B$.
\State Initialize $M^{\mathrm{pre}}\gets0$ and $\mathcal S\gets\emptyset$.
\For{each $\theta\in\Theta_\nu$ and each dyadic $B$-layer}
\State Add this threshold, layer, its actual window lengths, and initially
false mark bits to $\mathcal S$.
\State Set $k_A\gets|W_A|$ and
$q_\theta\gets\lceil C_qk_A^{1-\gamma}/\theta\rceil$.
\If{$q_\theta\ge k_A$}
\State Continue to the next threshold/layer; its all-false mark bits remain in
$\mathcal S$.
\EndIf
\State Set $r_\theta\gets\lceil2/\theta\rceil$ and
$g_\theta\gets\lceil C_h(k_A/q_\theta)\log n\rceil$.
\State Draw pools $\mathcal C_1,\ldots,\mathcal C_{r_\theta}$, each of $g_\theta$ windows sampled uniformly from $W_A$.
\For{each centre $w_c$ occurring in some pool}
\State For every $w_a\in W_B$, compute and cache the canonical alignment of
the unordered pair $\{w_a,w_c\}$ (lexicographically smallest among maximum
alignments), set both projections $\opt_{a,c},\opt_{c,a}$ from this same
cache entry, and record $\lcs(w_c,w_a)$; every later occurrence reuses it.
\State For every $w_b\in W_A$, build the threshold-dependent packing
$Y^{(\theta)}_{c,b}$ greedily until
$\lcs(w_c\setminus Y^{(\theta)}_{c,b},w_b)<d\theta^2/4$, and set
$\mathcal L^{(\theta)}_{c,x}\gets\{b:x\in Y^{(\theta)}_{c,b}\}$ for $x\in[d]$.
\EndFor
\For{each $w_a$ of the layer, writing $D_a\gets\max\{|w_a|,d\}$}
\For{$t=1,\ldots,r_\theta$}
\State Let $w_c$ be the first centre of $\mathcal C_t$ with $\lcs(w_c,w_a)\ge\theta D_a$ whose pair with $w_a$ is unmarked at $\theta$; if there is none, continue with the next stage.
\State Mark $(c,a)$ at $\theta$ and set $M^{\mathrm{pre}}_{ca}\gets\max\{M^{\mathrm{pre}}_{ca},\lcs(w_c,w_a)\}$.
\State Initialize $\mathsf{cnt}_b\gets0$ for every $b\in W_A$.
\State Set $p\gets\min\{1,2C_R\log n/(\theta^2D_a)\}$ and retain each $x\in\opt_{a,c}$ independently with probability $p$.
\State Abandon this stage if more than $4C_R\theta^{-2}\log n$ coordinates are retained.
\State For every retained $x$ traverse $\mathcal L^{(\theta)}_{c,x}$, incrementing $\mathsf{cnt}_b$ for every $b$ met.
\For{each $b$ with $\mathsf{cnt}_b\ge\lceil3p\theta^2D_a/16\rceil$}
\State $M^{\mathrm{pre}}_{ba}\gets\max\{M^{\mathrm{pre}}_{ba},\theta^4D_a/32\}$ and mark $(b,a)$ at $\theta$.
\EndFor
\EndFor
\EndFor
\State Store the final mark bits of this threshold/layer in $\mathcal S$.
\EndFor
\State Set $\mathcal W_\chi\gets(A_\chi,B_\chi,W_A,W_B,\text{layer maps},
\text{row/column maps})$.
\State Set $\mathcal I_\chi\gets(\mathcal W_\chi,d,k,\gamma,\chi,
\epsilon_{\mathrm{win}},\Theta_\nu,M^{\mathrm{pre}},\mathcal S)$.
\State \Return $\mathcal I_\chi$.
\EndProcedure
\end{algorithmic}
\end{algorithm}

\subsection{The global staircase cap}
\label{sec:bootstrap}

For a one-sided estimate matrix $N$, let $\operatorname{val}(N)$ denote the value
returned by the block dynamic program of \cite[Algorithm~6]{rsss19}. Thus
\cite[Lemma~4.8]{rsss19} uses only the numerical entries of $N$; it does not
require the corresponding common subsequences.

Fix one outer density guess $\nu$ and use the full, explicitly defined
net $\Theta_\nu\subseteq[\epsilon_{\mathrm{win}}\nu,1]$, independently
of which Step~1 marking computations are skipped. For a window pair $(i,j)$, write
\[
L_{ij}:=\lcs(w_i,w_j),
\qquad
D_{ij}:=\max\{|w_i|,|w_j|\}.
\]
All maxima over an empty set below are defined to be zero. The construction has three parts, taken in turn below: a global target matrix built from the threshold net, a repair phase that reaches that target on all but a small fraction of a fixed comparator path, and an induction that iterates the substitution at any fixed depth.

The next lemma packages all local-density thresholds into one one-sided analysis
matrix whose staircase value is globally large. It also shows that only a
controlled number of entries remain deficient relative to this target, which is
the input needed by the repair step.

\begin{lemma}[Global staircase cap]
\label{lem:staircase_cap}
Condition on a successful Step~1 run, and let $M$ be its one-sided
estimate matrix. Choose an orientation for which
Lemma~\ref{lem:rsss_window_compatible} supplies the required compatible
sequence. In this lemma let $n$ denote the common padded length, put
$\lambda^+:=\lcs(A,B)/n$, and let $\nu$ be the dyadic guess satisfying
$\nu\le\lambda^+<2\nu$. If $c_e>0$
is the approximation constant of the repair routine, put
\[
\beta:=\frac12\min\{1/32,c_e\},
\]
and define
\begin{align*}
C_{ij}
&:=
\beta D_{ij}
\max\{\theta^4:\theta\in\Theta_\nu,                 \theta\le L_{ij}/D_{ij}\},\\
M^\star_{ij}
&:=\max\{M_{ij},C_{ij}\}.
\end{align*}
Then $M^\star$ is one-sided and
\[
\operatorname{val}(M^\star)=\Omega(\beta\nu^4 n).
\]
Moreover, for
\[
\mathcal U:=\{(i,j):M_{ij}<C_{ij}\}
\]
one has
\[
|\mathcal U|
\le C_U k^{2-\gamma}\nu^{-1}\log^{q_U}n,
\qquad q_U=4,
\]
\end{lemma}

\begin{proof}
If $C_{ij}>0$, let $\theta_{ij}$ be the largest threshold attaining its
definition. Since $\theta_{ij}\le L_{ij}/D_{ij}\le1$ and $\beta\le1$,
\[
C_{ij}=\beta\theta_{ij}^4D_{ij}
\le \theta_{ij}D_{ij}
\le L_{ij}.
\]
Thus both $C$ and $M^\star$ are one-sided.

We next restate the part of the error-free RSSS argument that combines
different local densities. Lemma~\ref{lem:rsss_window_compatible}, with
$\epsilon_0=\epsilon_{\mathrm{win}}\nu$, gives a compatible sequence of
disjoint window pairs whose local LCS values $L_t$ satisfy
\[
\sum_tL_t\ge
c_{\mathrm{str}}\lambda^+n-2\epsilon_{\mathrm{win}}\nu n
\ge(c_{\mathrm{str}}-2\epsilon_{\mathrm{win}})\lambda^+n.
\]
For these pairs put $D_t$ equal to the larger window length and
$\rho_t=L_t/D_t$. Compatibility and disjointness give
\[
\sum_tD_t
\le
\sum_t(|w_{i_t}|+|w_{j_t}|)
\le2n.
\]
Pairs with $\rho_t<\epsilon_{\mathrm{win}}\nu$ contribute at most
$2\epsilon_{\mathrm{win}}\nu n\le2\epsilon_{\mathrm{win}}\lambda^+n$.
Since $\epsilon_{\mathrm{win}}\le c_{\mathrm{str}}/8$, the remaining pairs
have total LCS value at least
$(c_{\mathrm{str}}-4\epsilon_{\mathrm{win}})\lambda^+n
\ge c_{\mathrm{str}}\lambda^+n/2\ge c_{\mathrm{str}}\nu n/2$.
For each remaining
pair the multiplicative net contains a threshold $\theta_t$ with
$\theta_t\le\rho_t<(1+\epsilon_{\mathrm{win}})\theta_t$. Hence the compatible path using
these pairs has value at least
\[
\sum_t C_{i_tj_t}\ge\frac{\beta}{(1+\epsilon_{\mathrm{win}})^4}\sum_tD_t\rho_t^4\ge\frac{\beta(\sum_tD_t\rho_t)^4}{(1+\epsilon_{\mathrm{win}})^4(\sum_tD_t)^3}=\Omega(\beta\nu^4n).
\]
The first inequality follows from the definition of $C$ and the net relation
$\theta_t\le\rho_t<(1+\epsilon_{\mathrm{win}})\theta_t$: for each retained pair,
$C_{i_tj_t}\ge\beta D_t\theta_t^4\ge
\beta D_t\rho_t^4/(1+\epsilon_{\mathrm{win}})^4$. The second inequality is
the weighted Jensen inequality, in the form of \cite[Lemma~A.7]{rsss19} with
$y=3$. The third step uses
$\sum_tD_t\rho_t=\sum_tL_t\ge c_{\mathrm{str}}\nu n/2$ and
$\sum_tD_t\le2n$, proved immediately above, together with the absolute bound
$\epsilon_{\mathrm{win}}\le c_{\mathrm{str}}/8$. This is the Appendix-A
aggregation used in the proof of \cite[Theorem~2.1]{rsss19}, with the constant
$1/16$ replaced by $\beta$.
Since $M^\star\ge C$, its block-DP value satisfies the same bound.

Finally assign every pair in $\mathcal U$ to its largest qualifying
threshold $\theta_{ij}$. If Step~1 had marked this pair at that threshold, it would have written at least
$D_{ij}\theta_{ij}^4/32>C_{ij}$, contradicting $M_{ij}<C_{ij}$. It is
therefore a true $\theta_{ij}$-high false negative. For each fixed threshold, Lemma~\ref{lem:bipartite} when Step~1 is the oblivious procedure, and Lemma~\ref{lem:adaptive_step_one} when it is the adaptive one, applied in absolute units to all dyadic $B$-layers at once, bounds the number of such unmarked mixed pairs by
$\widetilde O(k^{2-\gamma}/\theta_{ij})$.  Its proof does not require the
$B$-windows to have one exact length. Summing over the polylogarithmic
threshold net and using $\theta_{ij}\ge\epsilon_{\mathrm{win}}\nu$ proves the asserted
bound.
\end{proof}

\begin{figure}[!ht]
\centering
\includegraphics[width=\linewidth]{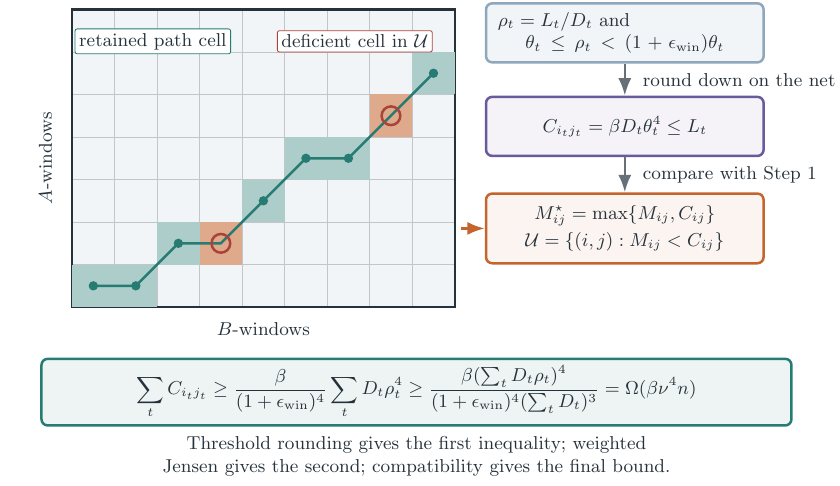}
\vspace{-1.5em}
\caption{Compatible staircase in the window matrix for the global staircase cap of Lemma~\ref{lem:staircase_cap}. Local densities on a compatible path are rounded down to the threshold net, producing one-sided targets $C_{ij}$; the orange cells represent the sparse deficient set $\mathcal U$ that remains for Step~2.}
\label{fig:global_staircase_cap}
\end{figure}

\subsection{Multiscale exact repair}
\label{sec:capped-repair}

The repair below keeps every LCS computation exact. Instead of opening one
neighborhood per trigger and repairing all of its cells, it alternates $s$
rounds of sampled exact detection with $s$ rounds of localization, shrinking
the candidate region at every scale, and repairs only the final region
exhaustively. No approximate routine is invoked anywhere, so the approximation
constant of this section is absolute.

The next lemma repairs this sparse deficient set by alternating sampled exact
detection with progressively smaller neighborhoods. It preserves a one-sided
global value while giving a running-time bound that remains uniform when the
number of scales grows.

\begin{lemma}[Multiscale exact repair]
\label{lem:capped_repair}
Continue in the setting of Lemma~\ref{lem:staircase_cap}, instantiated with
$c_e:=1$, so that $\beta=1/64$, and write
$m:=|\mathcal U|$ and
$\bar m:=C_Uk^{2-\gamma}\nu^{-1}\log^{q_U}n$, so that
$m\le\bar m$, where $C_U$ and $q_U=4$ are the absolute constants of
Lemma~\ref{lem:adaptive_step_one}. Let $s=s(n)\ge1$ be a number of scales with $s\le\log n$, put $\tau:=c_\tau\beta\nu^4/s$ for a sufficiently small absolute
constant $c_\tau>0$, and let radii $r_1\ge r_2\ge\cdots\ge r_s\ge1$ be given,
with $W_\ell:=r_\ell d$ and $h_\ell:=\lceil\tau r_\ell\rceil$. Assume $\tau r_s\ge2$, $w_{\max}\le W_s$, $2W_1<n$ and $r_1\le P_0r_s$. Then procedure
\textnormal{\textsc{MultiscaleStepTwo}} in Algorithm~\ref{alg:capped_repair}
returns, with high probability, a one-sided value
$Z_{\mathrm{out}}=\Omega(\beta\nu^4n)$ in time
\begin{align}
\label{eq:macro_step2}
\widetilde O(k^2P_0/(\tau r_1))
+\sum_{\ell=2}^{s}\widetilde O(\bar mP_0r_{\ell-1}/(\tau^2r_\ell))
+\widetilde O(\bar mP_0r_s/\tau)+\widetilde O(k^2),
\end{align}
where $P_0$ is the cost of testing one window pair exactly. Every factor
suppressed by $\widetilde O$ here and in the failure probability is
$s^{O(1)}(\log n)^{O(1)}$ with absolute constants, so the bound is uniform in
$s$.
\end{lemma}

\begin{algorithm}[!ht]
\caption{Multiscale exact repair (Step~2)}
\label{alg:capped_repair}
\begin{algorithmic}[1]
\Procedure{MultiscaleStepTwo}{$\mathcal I_\chi,\nu,r_1,\ldots,r_s$}
\Comment{Lemma~\ref{lem:capped_repair}}
\State Unpack $\mathcal I_\chi$ as defined in
Algorithm~\ref{alg:adaptive_step_one}, obtaining the oriented strings, both
window families and maps, $d,k,\gamma,\chi,\epsilon_{\mathrm{win}},
\Theta_\nu,M^{\mathrm{pre}}$, and $\mathcal S$.
\State Set $\beta\gets1/64$, $q_{\mathrm{bud}}\gets q_U+6$,
$\tau\gets c_\tau\beta\nu^4/s$, $h_\ell\gets\lceil\tau r_\ell\rceil$,
$W_\ell\gets r_\ell d$, and
$p_\ell\gets\min\{1,C_{\mathrm{hit}}\log n/(\tau r_\ell)\}$.
\State Set $\mathsf B_0\gets |W_A||W_B|$.
\State For $\ell=1,\ldots,s$, set
$\mathsf B_\ell,\mathsf G_\ell\gets
\lceil C_{\mathrm{bud}}sC_Uk^{2-\gamma}r_\ell\nu^{-1}
\tau^{-1}\log^{q_{\mathrm{bud}}}n\rceil$.
\State Set $\mathsf B'_\ell\gets
\lceil C_{\mathrm{bud}}sp_\ell\mathsf B_{\ell-1}\rceil$
for every $\ell=1,\ldots,s$.
\State Store the $A$-window starts in one increasing array and the
starts of every dyadic $B$-layer in its own increasing array.
\State To enumerate a $2W_\ell$ neighborhood, take its one contiguous
$A$-range and, for every dyadic $B$-layer, take that layer's contiguous
$B$-range and emit their Cartesian product; later deduplicate identifiers.
\For{$t=1,\ldots,\lceil C_{\mathrm{rep}}\log n\rceil$}
\State $Q_0\gets W_A\times W_B$, represented implicitly.
\For{$\ell=1,\ldots,s$}
\State Initialize $Q_\ell\gets\emptyset$, the exact-test counter
$u_\ell\gets0$, and the raw-generation counter $g_\ell\gets0$.
\State Set $k_A\gets|W_A|$ and sample exactly
$\min\{k_A,\lceil p_\ell k_A\rceil\}$ $A$-rows uniformly
without replacement, with fresh coins.
\For{each sampled $A$-row, and each pair $(i,j)$ in its
$Q_{\ell-1}$ bucket}
\State Increment $u_\ell$ and discard the trial immediately if
$u_\ell>\mathsf B'_\ell$.
\State Compute $D_{ij}:=\max\{|w_i|,|w_j|\}$ and
$L_{ij}$ exactly using Lemma~\ref{lem:perpair} on the windows identified by
$\mathcal I_\chi$.
\State From the full net and layer data stored in $\mathcal S$, set
$\theta_{ij}\gets\max\{\theta\in\Theta_\nu:
\theta\le L_{ij}/D_{ij}\}$ among thresholds applicable to this pair,
with maximum zero when the set is empty, and set
$C_{ij}\gets\beta D_{ij}\theta_{ij}^4$.
\If{$M^{\mathrm{pre}}_{ij}<C_{ij}$}
\For{each pair whose two starting positions lie within $2W_\ell$ of
those of $(i,j)$, enumerated across every dyadic $B$-layer as above}
\State Increment $g_\ell$ and discard the trial immediately if
$g_\ell>\mathsf G_\ell$; otherwise append the pair to $Q_\ell$.
\EndFor
\EndIf
\EndFor
\State Sort the generated pair identifiers lexicographically, delete duplicates,
store the resulting $Q_\ell$ in buckets indexed by its $A$-row, and discard
the trial if $|Q_\ell|>\mathsf B_\ell$.
\EndFor
\State For every $(i,j)\in Q_s$ compute $L_{ij}$ exactly on its two
orientation-specific windows; let $M^{t}$ agree with $M^{\mathrm{pre}}$ off
$Q_s$ and equal $\max\{M^{\mathrm{pre}}_{ij},L_{ij}\}$ on $Q_s$.
\EndFor
\State \Return the maximum of $\operatorname{val}(M^{t})$ over all completed
trials, and $\operatorname{val}(M^{\mathrm{pre}})$ if none completed.
\EndProcedure
\end{algorithmic}
\end{algorithm}

\begin{proof}
Condition throughout on the joint Step~1 event of
Lemma~\ref{lem:adaptive_step_one}, on which $M^{\mathrm{pre}}$ is one-sided
and all residual bounds hold. Every entry of every completed-trial matrix
$M^t$ is then an entry of $M^{\mathrm{pre}}$ or an exact window LCS, so every
completed-trial value is one-sided.

For the analysis of one trial, define an uncapped process
$\widetilde Q_0,\ldots,\widetilde Q_s$ that uses exactly the same row samples,
exact tests, and neighborhood expansions as Algorithm~\ref{alg:capped_repair},
but never aborts because of a budget.  The implemented process follows this
uncapped process until the first budget violation.

\emph{Approximation.} Fix the comparator path of the successful orientation and
discard its cells of local density below $\epsilon_{\mathrm{win}}\nu$; as in
Lemma~\ref{lem:staircase_cap}, the surviving cells satisfy
$\sum_tC_{i_tj_t}=\Omega(\beta\nu^4n)$, and every individual cap value is at
most $d$. Order the deficient surviving cells along the path. For $\ell\in[s]$, call such a cell \emph{$\ell$-lonely} when fewer than $h_\ell$ earlier deficient cells of this order have starting position within $W_\ell$ of its own in both coordinates. The earlier cells within $W_\ell$ in one fixed coordinate form a suffix of the order, because the path is monotone in both coordinates, so the cells within $W_\ell$ in both coordinates form the shorter of these two suffixes, and a cell is $\ell$-lonely exactly when one of the two suffixes has fewer than $h_\ell$ cells. Partition each string into at most $\lceil n/W_\ell\rceil$
intervals of length at most $W_\ell$.  In either coordinate, after the first
$h_\ell$ deficient cells met in one interval, every later cell in that
interval has $h_\ell$ earlier cells at distance at most $W_\ell$.  Hence the
union of the two coordinate-lonely sets has size at most
$2h_\ell\lceil n/W_\ell\rceil\le8\tau n/d$: here
$W_\ell<n$, $h_\ell\le2\tau r_\ell$ follows from
$\tau r_\ell\ge2$, and
$\lceil n/W_\ell\rceil\le2n/W_\ell$.  Therefore the cells lonely at some
level carry total cap value at most
$8s\tau n=8c_\tau\beta\nu^4n$.  Taking the absolute constant
$c_\tau$ small enough makes this at most half of the staircase sum.

\emph{Coverage.} Consider a surviving deficient cell $z$ that is $\ell$-lonely
for no $\ell\in[s]$. We show by induction on $\ell$ that, with probability at least
$1-\ell n^{-C_{\mathrm{hit}}}$, the region $\widetilde Q_\ell$ contains $z$; moreover,
when $\ell<s$, it contains every deficient cell of the order whose starting
positions are within $W_{\ell+1}$ of those of $z$ in both coordinates. Since $z$ is not $\ell$-lonely, at least $h_\ell$ earlier deficient cells lie within $W_\ell$ of $z$ in both coordinates.
These witnesses occupy distinct $A$-rows, because compatible path cells use
distinct $A$-windows. Write $k_A:=|W_A|$. Each witness lies in
$\widetilde Q_{\ell-1}$: for $\ell=1$ this
is trivial, and for $\ell\ge2$ it is the induction hypothesis, because
$W_\ell\le W_{\ell-1}$. If $p_\ell=1$, every $A$-row is sampled and the probability of missing all
$h_\ell$ witnesses is zero. If $p_\ell<1$, sampling
$\lceil p_\ell k_A\rceil$ rows without replacement misses them with probability
at most
$(1-h_\ell/k_A)^{p_\ell k_A}\le\exp(-p_\ell h_\ell)\le
n^{-C_{\mathrm{hit}}}$, since $h_\ell\ge\tau r_\ell$. Thus in either
case a sampled witness $w_\ell$ is tested, because it lies in $\widetilde Q_{\ell-1}$, and is
found deficient, and $\widetilde Q_\ell$ then receives every pair whose starting positions
are within $2W_\ell$ of those of $w_\ell$. This includes $z$, since $w_\ell$ is within
$W_\ell\le2W_\ell$ of it. When $\ell<s$, it also includes every deficient
cell within $W_{\ell+1}$ of $z$, by the triangle inequality
$W_{\ell+1}+W_\ell\le2W_\ell$. After level $s$ the cell $z$ lies in $\widetilde Q_s$ and is repaired exactly, so $M^t$ carries its full
value $L_z\ge C_z$ there. A union bound over the at most $k^2$ path cells
bounds the total coverage failure by $k^2sn^{-C_{\mathrm{hit}}}$.

Let $\widetilde M^t$ be the final table of the uncapped process.  On
this coverage event every surviving path cell is nondeficient, lonely at some
scale, or repaired, and $\widetilde M^t\ge M^{\mathrm{pre}}$ entrywise.
Consequently
\begin{align*}
\operatorname{val}(\widetilde M^t)
\ge\sum_tC_{i_tj_t}-8s\tau n
=\Omega(\beta\nu^4n).
\end{align*}

\emph{Running time and coupling.} Write $k_A:=|W_A|$ and
$k_B:=|W_B|$.  If $p_\ell<1$, then
$p_\ell=C_{\mathrm{hit}}\log n/(\tau r_\ell)$ and
$2W_\ell\le2W_1<n$ imply $r_\ell<k_A/2$ up to the fixed endpoint-rounding
factor.  Since $\tau\le1$, we have
$p_\ell k_A\ge2C_{\mathrm{hit}}\log n$ for all sufficiently large $n$.
Consequently
\begin{align*}
\frac{\lceil p_\ell k_A\rceil}{k_A}\le2p_\ell .
\end{align*}
For $p_\ell=1$ this inclusion probability is exactly one.  Thus level $1$
tests $\widetilde O(p_1k^2)$ pairs and costs
$\widetilde O(k^2P_0/(\tau r_1))$.

Analyze first the uncapped process.  Only members of $\mathcal U$ can trigger,
and Observation~\ref{obs:nearby-counts} gives
$\widetilde O(r_\ell^2)$ output pairs per trigger because
$W_\ell\ge W_s\ge w_{\max}$.  The ordered start arrays stated in the algorithm enumerate precisely these
pairs without scanning the full product: one contiguous $A$-range is paired
with the contiguous range from each dyadic $B$-layer. The number of layers is
$O(\log n)$ and is absorbed by $\widetilde O$.
Let $\widetilde R_\ell$ be the raw identifier multiset emitted before
deduplication.  The preceding inclusion bound gives
\begin{align*}
\E[|\widetilde R_\ell|]
 =\widetilde O(p_\ell mr_\ell^2)
 =\widetilde O(mr_\ell/\tau),
\qquad
\E[|\widetilde Q_\ell|]
 \le\E[|\widetilde R_\ell|].
\end{align*}
Sorting $g$ identifiers, deleting repeats, and forming row buckets costs
$O(g\log g)$ time.  Using $m\le C_Uk^{2-\gamma}\nu^{-1}\log^{q_U}n$ and the two
ordered-range searches of Observation~\ref{obs:nearby-counts}, every omitted
counting factor is at most $\log^4n$. Thus the executable choice
$q_{\mathrm{bud}}=q_U+6$ dominates raw generation, deduplicated size, the
layer lookup and sorting. Consequently $\mathsf G_\ell$ and
$\mathsf B_\ell$ exceed their respective expectations by
$\Omega(C_{\mathrm{bud}}s)$; the budgets are computable without knowing
$\mathcal U$.

At level $\ell\ge2$, expose the budget events sequentially.  Conditioned on
all earlier levels respecting their budgets,
$|\widetilde Q_{\ell-1}|\le\mathsf B_{\ell-1}$, and the fresh level-$\ell$
row sample is independent of $\widetilde Q_{\ell-1}$.  The expected number
of exact tests is therefore at most
$2p_\ell\mathsf B_{\ell-1}$; level $1$ satisfies the same bound
deterministically.  After increasing the absolute budget constant to absorb
this factor two, $\mathsf B'_\ell$ exceeds the conditional expectation by
$\Omega(C_{\mathrm{bud}}s)$.  Markov's inequality, applied successively, and
a union bound over raw generation, deduplicated size, and exact tests at all
$s$ levels show that the uncapped process respects every budget with
probability at least $1-3/C_{\mathrm{bud}}\ge5/6$.

Couple Algorithm~\ref{alg:capped_repair} to this uncapped process by using the
same fresh row samples and the same deterministic test and enumeration order.
Induction on $\ell$ shows that, on the budget-good event, no abort occurs and
$Q_\ell=\widetilde Q_\ell$ at every level.  Thus the implemented final table
equals $\widetilde M^t$ on the intersection of the budget-good and coverage
events.
On a
completed trial, the level-$\ell$ test cost is deterministically at most
$\mathsf B'_\ell P_0=
\widetilde O_s(\bar m r_{\ell-1}P_0/(\tau^2r_\ell))$ for $\ell\ge2$.
The final scan costs
$|Q_s|P_0=\widetilde O(\bar m r_sP_0/\tau)$. On a completed trial,
candidate generation and deduplication touch at most
$\sum_{\ell=1}^s\mathsf G_\ell=
\widetilde O_s(\bar m r_1/\tau)$ identifiers, at polylogarithmic cost each.
Since $r_1\le P_0r_s$, this work is at most the final-scan term. When $p_\ell=1$ the
level tests all of $Q_{\ell-1}$, and $\tau r_\ell=O(\log n)$ then makes this
cost $\widetilde O(\bar m r_{\ell-1}P_0/(\tau^2r_\ell))$ as well. A trial therefore
completes with all coverage events with probability at least $1/2$, so $O(\log n)$ independent trials succeed with high probability; the returned maximum is at least the value of a successful trial, and it is one-sided because every trial's value is. The block dynamic program costs $\widetilde O(k^2)$ per trial, which is the last term of Eq.~\eqref{eq:macro_step2}.
Summing the level costs gives Eq.~\eqref{eq:macro_step2}. Finally, $q_U=4$ and $q_{\mathrm{bud}}=q_U+6$ are absolute. The
trial repetition contributes one further $\log n$, each range search and
sort contributes at most $\log^2n$, and the threshold, layer, orientation
and density-guess loops together contribute at most $\log^4n$. Thus all
work omitted from Eq.~\eqref{eq:macro_step2} is bounded explicitly by
$O(s^4\log^{q_{\mathrm{bud}}+8}n)$ times an absolute constant. There is
no recursive call and no other dependence on $s$, proving the stated
uniformity for $s\le\log n$.
\end{proof}

\begin{figure}[!ht]
\centering
\includegraphics[width=\linewidth,trim=0 8pt 0 0,clip]{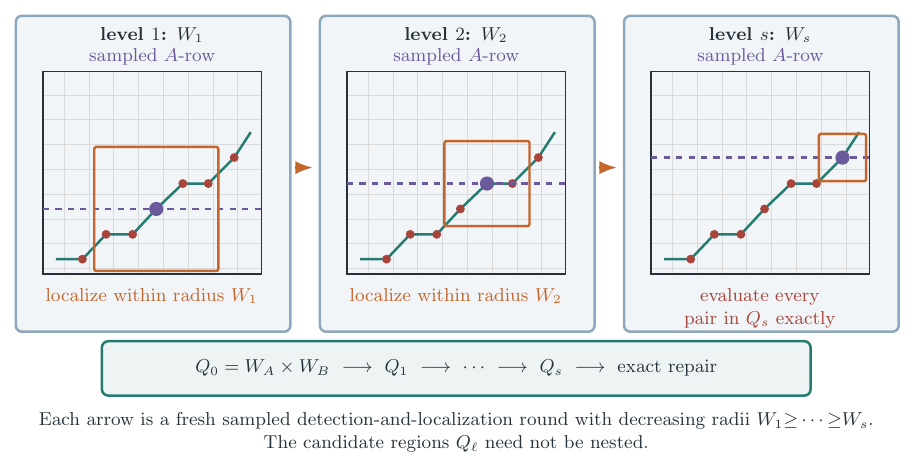}\setlength{\abovecaptionskip}{2pt}
\caption{The multiscale exact repair of Lemma~\ref{lem:capped_repair}. At each level a fresh row sample detects a deficient witness and opens a smaller spatial neighborhood; all pairs in the final region $Q_s$ are evaluated exactly. The decreasing radii do not assert set inclusion among the candidate regions $Q_\ell$.}
\label{fig:multiscale_exact_repair}
\end{figure}

\subsection{The multiscale exponents}
\label{sec:iterate}

The next lemma chooses the radii and window parameters so that preprocessing,
filtering, and repair have one controlled exponent. It yields the uniform
exponent $F_{s,a}$, which approaches $8/5$ when the scale count grows and the
crossover parameter shrinks.

\begin{lemma}[Adaptive one-level exponent]
\label{lem:one_level_adaptive}
Let $s=s(n)$ be an integer with $2\le s\le\sqrt{\log n}$, and let $0<a<2(s+1)/((s-1)(30s+3))$, for which $a\le1/(33s)$ is a convenient sufficient form. Put
$x:=(4s+2+3(s-1)a)/(5s+2)$, let $\gamma$ solve $x\gamma=3x-2+3a$, and set
$k=n^x$, $d=n^{1-x}$. For a density guess $\nu$, define the ideal radii
\begin{equation}
 r_\ell^\star:=
 k^{(s-\ell+1)\gamma/(s+1)}
 \nu^{(5s+5-9\ell)/(s+1)},
 \qquad \ell=1,\ldots,s. \label{eq:adaptive_radii}
\end{equation}
Use the upward-rounded monotone radii
$r_s:=\max\{1,\lceil r_s^\star\rceil\}$ and
$r_\ell:=\max\{r_{\ell+1},\lceil r_\ell^\star\rceil\}$ for
$\ell=s-1,\ldots,1$. Running Lemma~\ref{lem:adaptive_step_one} in place of
\textnormal{\textsc{WindowStepOne}} and then the $s$-scale repair of
Lemma~\ref{lem:capped_repair}, every density guess $\nu=n^{-b}$ with $b\le a$ is handled in time $s^{O(1)}(\log n)^{O(1)}n^{F_{s,a}}$ with absolute constants, hence uniformly in $s$, where
\begin{align*}
F_{s,a}=\frac{8s+4+21sa}{5s+2}=\frac85+\frac{4}{25s+10}+\frac{21s}{5s+2}a,
\end{align*}
and all hypotheses of Lemma~\ref{lem:capped_repair} hold.
\end{lemma}
\begin{proof}
The three leading costs are $k^{1+\gamma}d^2$ from centre preprocessing,
$k^2\nu^{-3}$ from the filters, and
$d^2k^{2-c_s\gamma}\nu^{-9c_s}+k^2$ from the repair, with exponents
$t_1=2-x+x\gamma$, $t_2(b)=2x+3b$ and $t_3(b)=2-c_sx\gamma+9c_sb$, where
$c_s=s/(s+1)$; the last two increase in $b$, so $b=a$ is the worst guess.
The stated $x$ and $\gamma$ are exactly the solution of
$t_1=t_2(a)=t_3(a)$, and the common value is $2x+3a=F_{s,a}$. The remaining
preprocessing term $k^\gamma n/\nu$ has exponent $1+x\gamma+b$, below $t_1$
by $1-x-b>0$. Since $1-x=(s-3(s-1)a)/(5s+2)$ and
$x\gamma-9a=(s+1)(2-21a)/(5s+2)$, the assumption on $a$, which forces $a<2/21$, and $s\ge2$ give $0<\gamma<1$, $b\le a<1-x$, hence $w_{\max}\le d^2$ and
$P_0=\widetilde O(d^2)$; $x\gamma>9a\ge9b$, hence the ideal radii satisfy
$r_\ell^\star/r_{\ell+1}^\star=(k^\gamma\nu^9)^{1/(s+1)}$, which is at least $n^{(2-21a)/(5s+2)}$. Hence the ideal radii decrease and, since
$\tau=c_\tau\beta\nu^4/s$,
$\tau r_s^\star\ge n^{(2-21a)/(5s+2)}/O(s)\to\infty$ under
$s\le\sqrt{\log n}$. Moreover
\[
r_s^\star\nu=\frac{s}{c_\tau\beta\nu^3}\tau r_s^\star
\ge\frac{s}{c_\tau\beta}\tau r_s^\star\to\infty.
\]
The equality is the definition of $\tau$, the inequality uses $\nu\le1$,
and the final limit follows from the preceding bound on $\tau r_s^\star$.
Thus, after enlarging the uniform input threshold,
$r_s^\star\nu\ge1/\epsilon_{\mathrm{win}}$ and therefore
$W_s^\star:=dr_s^\star\ge d/(\epsilon_{\mathrm{win}}\nu)=w_{\max}$; and
\[
r_1^\star=k^{c_s\gamma}\nu^{(5s-4)/(s+1)},\qquad
W_1^\star:=d r_1^\star.
\]
Since $(5s-4)/(s+1)>0$ and $\nu\le1$, we have
$r_1^\star\le k^{c_s\gamma}$ and hence
$W_1^\star\le n^{1-x+c_sx\gamma}$. Moreover
$1-(1-x+c_sx\gamma)=x(1-c_s\gamma)\ge x/(s+1)$ because $\gamma<1$. Thus $n/W_1^\star\ge n^{x/(s+1)}\to\infty$ under $s\le\sqrt{\log n}$, and in particular $2W_1^\star<n$ for all sufficiently large $n$. Finally $r_1^\star/r_s^\star=n^{(s-1)(x\gamma-9b)/(s+1)}$ while $P_0=n^{2(1-x)}$, and $2(1-x)-(s-1)x\gamma/(s+1)$ has the sign of $2-a(s-1)(30s+3)/(s+1)$, which is positive precisely by the assumption on $a$; hence $r_1^\star\le P_0r_s^\star$ at every $b\le a$, this being the condition that couples $a$ to $s$. The block dynamic program is the explicit $k^2$ term of Lemma~\ref{lem:capped_repair}, dominated by $k^2\nu^{-3}$.

For the actual algorithm take $d=2^{\lceil(1-x)\log_2n\rceil}$ and perform
the single padding step of Algorithm~\ref{alg:approximate_lcs}. Since
$d_{\max}=O(n^{1-x+a})=o(n)$, its common padded length satisfies
$n\le n^+<2n$; hence $k:=n^+/d=\Theta(n^x)$, differing by only a constant
factor from the unpadded surrogate $\lceil n/d\rceil$. Use the dyadic guesses
of Algorithm~\ref{alg:approximate_lcs} and the rounded radii defined in the
statement. Every change from an ideal radius is by at most a constant factor.
The strict margins above are $n^{\Omega(1/s)}$ and
$s\le\sqrt{\log n}$, so after this rounding
$r_s\nu\ge1/\epsilon_{\mathrm{win}}$, hence
$W_s\ge w_{\max}$, while $2W_1<n$, $\tau r_s\ge2$, and
$r_1\le P_0r_s$ still hold for all sufficiently large $n$.  Ceilings in
sample counts have already been charged in Lemma~\ref{lem:capped_repair}.
\end{proof}

\begin{corollary}[Explicit three-scale exponent]
\label{cor:explicit_three_scale}
Running Lemma~\ref{lem:one_level_adaptive} with $s=3$ and $a=4/127$, and the
estimate of Corollary~\ref{cor:r} below the crossover, computes an
$\Omega(\lambda^3)$-approximation of $\lcs(A,B)$ in time
$\widetilde O(n^{224/127})$, where $224/127<1.7638$. It avoids the sublinear LIS estimator of Lemma~\ref{lem:anss}, using below the crossover only the refereed tradeoff of Lemma~\ref{lem:bcd}, and the window branch performs no rectangular matrix multiplication.
\end{corollary}
\begin{proof}
Here $x=106/127$, $\gamma=38/53$ and $d=n^{21/127}$, and $a=4/127$ is the
solution of $F_{3,a}=2-15a/2$. At the worst window guess $b=a$ the radii are
$r_1=n^{46/127}$, $r_2=n^{36/127}$ and $r_3=n^{26/127}$, and the exponents of
all seven cost lines, namely $k^{1+\gamma}d^2$, $k^2\nu^{-3}$,
$k^2d^2/(\tau r_1)$, $md^2r_1/(\tau^2r_2)$, $md^2r_2/(\tau^2r_3)$,
$md^2r_3/\tau$ and $n^{2-15a/2}$, are all equal to $224/127$; the block
dynamic program contributes only $k^2=n^{212/127}$ and the remaining
preprocessing only $k^\gamma n/\nu=n^{207/127}$. The hypothesis of
Lemma~\ref{lem:one_level_adaptive} reads $a<4/93$ at $s=3$ and holds at
$a=4/127$, as does $a\le2/15$ for Corollary~\ref{cor:r}; the binding side
condition $r_1\le P_0r_s$ is tightest at $b=0$, where
$r_1/r_s=n^{38/127}\le d^2=n^{42/127}$. This choice of $s$ is optimal: the
balance is feasible only for $s\in\{2,3\}$, giving $39/22$ and $224/127$,
while for $s\ge4$ the constraint binds before the balance and the exponent
grows again, to $221/123$ at $s=4$.
\end{proof}

\subsection{Proof of the main theorem}
\label{sec:main-proof}

\paragraph{Hardcoded constants and exact arithmetic.}
Fix a concrete implementation of Lemma~\ref{lem:bcd}. Thus there are absolute
constants $K_{\mathrm{BCD}},q_{\mathrm{BCD}}$ such that, at parameter
$\varepsilon_{\mathrm{BCD}}=30/127$, its one-sided output is at least
\[
\frac{\lcs(A,B)}
 {K_{\mathrm{BCD}}n^{12/127}\log^{q_{\mathrm{BCD}}}n}
\]
with the stated high probability. Fix
$c_{\mathrm{BCD}}>q_{\mathrm{BCD}}/3+2$, the structural constant
$c_{\mathrm{str}}$, a reciprocal power of two
$\epsilon_{\mathrm{win}}\le c_{\mathrm{str}}/8$, $\beta=1/64$, and the
absolute constants required by Lemmas~\ref{lem:adaptive_step_one} and
\ref{lem:capped_repair}. Use the explicit constant $K_{\mathrm A}=2$ from
Theorem~\ref{thm:anss_lcs_endpoint}. For a requested failure exponent $C$, the sampling
constants are then chosen in the order described in
Remark~\ref{rem:good_event}; $n_0(C)$ absorbs the finitely many remaining
small inputs and constant-factor margins. Fix $K_\star:=16$, which also
absorbs the fixed change of base between logarithms below.

Every exponent used below is rational. In Algorithm~\ref{alg:approximate_lcs}
its numerator and denominator are determined by the integer
$s=O(\sqrt{\log\log n})$. For a quantity $n^{u/v}$ used as an integral length or radius, the algorithm
takes the least power of two $2^j$ satisfying $2^{vj}\ge n^u$, tested by
integer arithmetic, and then rounds any remaining count upward. This changes a
quantity by at most a constant factor. For a mixed power
$k^{p/q}\nu^{r/q}$ with $\nu=2^{-j}$, it analogously takes the least power
$2^t$ satisfying $2^{qt+jr}\ge k^p$. The comparison is performed after
clearing negative exponents, using repeated squaring and
$s^{O(1)}$ words of $O(\log n)$ bits. Write
$\lg n:=\lceil\log_2n\rceil$.
Every Bernoulli probability $p$ is replaced by
$p^+:=\min\{1,2^{-B}\lceil2^Bp\rceil\}$, where
$B=B_0\lg n$ and the fixed constant $B_0$ exceeds all polynomial exponents in
the proof. Since every nonzero probability used here is $n^{-O(1)}$,
$p^+/p=1+o(1)$; the one-sided increase only improves hitting events and changes
work bounds by $1+o(1)$. Exactly $B$ random bits implement the dyadic trial.
Thresholds use $p^+$, and all comparisons are made after clearing their
$O(\log n)$-bit denominators. These conventions give one uniform word-RAM
program and require no real-RAM operations.

\begin{algorithm}[!ht]
\caption{Uniform $8/5+o(1)$ approximation for LCS}
\label{alg:approximate_lcs}
\begin{algorithmic}[1]
\Procedure{ApproximateLCS}{$A,B\in\Sigma^n,\ C>0$}
\Comment{Theorem~\ref{thm:uniform_eight_fifths_formal}}
\If{$n<n_0(C)$}
  \State \Return $\lcs(A,B)$ exactly.
\EndIf
\State Coordinate-compress the alphabet as in Section~\ref{sec:basic_notation}.
Set $H\gets\lceil\log_2(\lg n)\rceil$ using integer bit lengths and
$s\gets\lfloor\sqrt H/(132K_\star)\rfloor$ using integer square root.
\If{$s<2$}
  \State \Return $\lcs(A,B)$ exactly.
\EndIf
\State Set $a\gets1/(66s)$,
$x\gets(4s+2+3(s-1)a)/(5s+2)$, and
$\gamma\gets(3x-2+3a)/x$.
\State Set $d\gets2^{\lceil(1-x)\log_2n\rceil}$ and
$j_{\max}\gets\lfloor a\log_2n\rfloor$.
\State Set $\mathcal G\gets\{2^{-j}:0\le j\le j_{\max}\}$ and
$\nu_{\min}\gets2^{-j_{\max}}$.
\State Set $d_{\max}\gets d/(\epsilon_{\mathrm{win}}\nu_{\min})$.
Right-pad the original $A,B$ once, using two private dummy alphabets, to the
least common length $n^+$ divisible by $d_{\max}$; call the padded strings
$A^+,B^+$ and set $k\gets n^+/d$.
\State $Z\gets\Call{QuantitativeSmallLambdaEstimate}{A,B,C+1}$.
\Comment{Algorithm~\ref{alg:quantitative_small_lambda_estimate}}
\For{each $\nu\in\mathcal G$ and each of the four orientations $\chi$ of
the common padded pair $(A^+,B^+)$}
\State $\mathcal I_\chi\gets\Call{AdaptiveStepOne}{A^+,B^+,\nu,d,k,
\gamma,\chi,\epsilon_{\mathrm{win}}}$.
\Comment{Algorithm~\ref{alg:adaptive_step_one}}
\State Set $r_{s+1}\gets1$.
\For{$\ell=s,s-1,\ldots,1$}
  \State Set $r_\ell\gets\max\{r_{\ell+1},
  \lceil k^{(s-\ell+1)\gamma/(s+1)}
  \nu^{(5s+5-9\ell)/(s+1)}\rceil\}$.
\EndFor
\State $Z\gets\max\{Z,
\Call{MultiscaleStepTwo}{\mathcal I_\chi,\nu,r_1,\ldots,r_s}\}$.
\Comment{Algorithm~\ref{alg:capped_repair}}
\EndFor
\State \Return $Z$.
\EndProcedure
\end{algorithmic}
\end{algorithm}

\begin{remark}[The global good event]
\label{rem:good_event}
Fix $C$. First union-bound the Step~1 pool-hitting and sparse-filter events of
Lemma~\ref{lem:adaptive_step_one} over the polynomially many anchors, stages,
thresholds, layers, density guesses and four orientations. Condition on their
intersection. For each fixed guess and orientation, now fix the deterministic
canonical comparator path supplied by the input and the cached tie-breaking
rule. The row samples used by Lemma~\ref{lem:capped_repair} are fresh and
independent of both Step~1 and this fixed path. Conditional on the preceding
choices, its witness-hitting and budget bounds therefore apply to every cell of
that path. Union-bound these events over all path cells, all $s$ scales, guesses
and orientations, and amplify the constant-success repair trials by
$O(C\log n)$ independent repetitions. Finally include the high-probability
event of the independent small-density call:
Algorithm~\ref{alg:quantitative_small_lambda_estimate} for
Algorithm~\ref{alg:approximate_lcs}, or BCAD for
Theorem~\ref{thm:explicit_anss_free_formal}. Choosing the displayed sampling constants
in terms of $C$ makes the total failure probability at most $n^{-C}$ and changes
only polylogarithmic factors. This probability bound is uniform over all input
densities; the successful guess and comparator path may, and generally do,
depend on the input.
\end{remark}

The next theorem records the strongest explicit bound obtained here without
invoking the ANSS estimator. It combines the three-scale window branch with the
BCAD small-density branch to give a one-sided $\lambda^3$-scaled estimate in time
$\widetilde O_C(n^{224/127})$.

\begin{theorem}[Explicit ANSS-free bound]
\label{thm:explicit_anss_free_formal}
There is an absolute constant $c>0$ such that, for every fixed $C>0$, a
randomized algorithm, given strings $A,B$ of length $n$ and
$\lambda=\lcs(A,B)/n$, returns a number $Z$ satisfying
\[
c\lambda^3\lcs(A,B)\le Z\le\lcs(A,B)
\]
with probability at least $1-n^{-C}$, in time
$\widetilde O_C(n^{224/127})$.
\end{theorem}
\begin{proof}
Use the fixed parameters
\[
s=3,\qquad a=4/127,\qquad x=106/127,
\qquad \gamma=38/53.
\]
Set $d$ to the rounded value of $n^{21/127}$ and put
\[
j_{\max}:=\lceil a\log_2n+
c_{\mathrm{BCD}}\log_2\log_2n\rceil,
\quad
\mathcal G:=\{2^{-j}:0\le j\le j_{\max}\},
\quad
\nu_{\min}:=2^{-j_{\max}}.
\]
Pad once using $d_{\max}:=d/(\epsilon_{\mathrm{win}}\nu_{\min})$ exactly as
in Algorithm~\ref{alg:approximate_lcs}, producing $A^+,B^+$ and
$k:=n^+/d$. Initialize $Z$ by running the algorithm of
Lemma~\ref{lem:bcd} with $\varepsilon_{\mathrm{BCD}}=30/127$ and failure
exponent $C+1$.
For every $\nu\in\mathcal G$ and orientation $\chi$, run
Algorithm~\ref{alg:adaptive_step_one}, set
\[
r_3:=\max\{1,\lceil k^{\gamma/4}\nu^{-7/4}\rceil\},\quad
r_2:=\max\{r_3,\lceil r_3(\nu^4r_3)\rceil\},\quad
r_1:=\max\{r_2,\lceil r_3(\nu^4r_3)^2\rceil\},
\]
run Algorithm~\ref{alg:capped_repair} with these three radii, and update $Z$
by the maximum. Return
the final $Z$.

The BCAD call has approximation factor
$K_{\mathrm{BCD}}n^{12/127}\log^{q_{\mathrm{BCD}}}n$. By
Lemma~\ref{lem:bcd}, its runtime is
\[
O(n^{2-30/127})=O(n^{224/127}).
\]
Moreover, $\nu_{\min}\le
n^{-4/127}\log^{-c_{\mathrm{BCD}}}n$, so whenever
$\lambda<4\nu_{\min}$ we have
\[
\lambda^3\le64n^{-12/127}\log^{-3c_{\mathrm{BCD}}}n.
\]
The choice of $c_{\mathrm{BCD}}$ therefore makes the BCAD output at least
$c_0\lambda^3\lcs(A,B)$ for an absolute $c_0>0$ and all sufficiently large
$n$.

Suppose instead that $\lambda\ge4\nu_{\min}$. Padding does not change
$L:=\lcs(A,B)$ and satisfies $n\le n^+<2n$ because
$d_{\max}=\widetilde O(n^{25/127})<n$. Hence
$\lambda^+:=L/n^+\ge\lambda/2\ge2\nu_{\min}$, so the dyadic grid contains
a guess $\nu$ with $\nu\le\lambda^+<2\nu$. The pad-first convention and
Lemma~\ref{lem:rsss_window_compatible} give a successful one of the four
orientations. Lemmas~\ref{lem:adaptive_step_one},
\ref{lem:capped_repair}, and~\ref{lem:one_level_adaptive}, specialized to
$s=3$, $a=4/127$, $x=106/127$, and $\gamma=38/53$, then return at least
$c_1\nu^4n^+$ for an absolute $c_1>0$. Since
$\nu\ge\lambda^+/2\ge\lambda/4$ and $n^+\ge n$, this is at least
$(c_1/256)\lambda^4n=(c_1/256)\lambda^3L$.

Corollary~\ref{cor:explicit_three_scale} proves that the worst unshifted window
guess costs $\widetilde O(n^{224/127})$. Extending the grid by
$c_{\mathrm{BCD}}\log_2\log_2n$ dyadic levels changes each occurrence of
$\nu^{-O(1)}$ by only a polylogarithmic factor. All strict side conditions in
Lemma~\ref{lem:one_level_adaptive} retain polynomial slack, so they remain
valid. The global event is justified in Remark~\ref{rem:good_event}.

The BCAD value is one-sided. On the Step~1 and repair good events, every table
entry is either a certified lower bound or an exact window LCS, so every
compatible-path value is also one-sided. Private padding symbols never match.
Thus the maximum returned by the algorithm never exceeds $L$. Taking
$c:=\min\{c_0,c_1/256\}$ completes the proof; inputs below $n_0(C)$ are
solved exactly.
\end{proof}

The next theorem fixes the scale count as a function of an arbitrary accuracy
parameter. Combining the resulting window branch with the quantitative ANSS
endpoint gives a one-sided $\lambda^3$-scaled estimate in time
$O_{\delta,C}(n^{8/5+\delta})$.

\begin{theorem}[ANSS-assisted fixed-accuracy bound]
\label{thm:fixed_delta_formal}
There is an absolute constant $c>0$ such that, for every fixed
$\delta,C>0$, a randomized algorithm, given $A,B\in\Sigma^n$ and
$\lambda=\lcs(A,B)/n$, returns a number $Z$ satisfying
\[
 c\lambda^3\lcs(A,B)\le Z\le\lcs(A,B)
\]
with probability at least $1-n^{-C}$, in time
$O_{\delta,C}(n^{8/5+\delta})$.
\end{theorem}
\begin{proof}
Put $\delta_0:=\min\{\delta,1\}$ and choose the fixed constants
\[
 s:=\max\{2,\lceil1/\delta_0\rceil\},
 \qquad a:=\frac1{66s}.
\]
Then $a\le1/(33s)$, so Lemma~\ref{lem:one_level_adaptive} applies, and
\[
 F_{s,a}-\frac85=\frac4{25s+10}+\frac{21s}{5s+2}a\le\frac4{25s}+\frac7{110s}=\frac{123}{550s}<\frac{\delta_0}{4}.
\]
Because $s$ is fixed, the factors $s^{O(1)}(\log n)^{O(1)}$, the density
grid, and the four orientations are at most $n^{\delta_0/2}$ for all
sufficiently large $n$.  Thus the complete window branch takes
$O_{\delta,C}(n^{8/5+3\delta_0/4})$, which is
$O_{\delta,C}(n^{8/5+\delta})$.

Let
\[
 j_{\max}:=\lfloor a\log_2n\rfloor,
 \qquad \mathcal G:=\{2^{-j}:0\le j\le j_{\max}\},
 \qquad \nu_{\min}:=2^{-j_{\max}}.
\]
The use of the floor ensures
$n^{-a}\le\nu_{\min}<2n^{-a}$ and ensures that every
$\nu=n^{-b}$ in the grid has $b\le a$, as required by
Lemma~\ref{lem:one_level_adaptive}.  Run its $s$-scale window construction
for all guesses and all four orientations. Independently run
Algorithm~\ref{alg:quantitative_small_lambda_estimate}, the endpoint
algorithm of Theorem~\ref{thm:anss_lcs_endpoint}, and return the maximum of
all resulting values.  Choose the sampling constants in both branches with
failure exponent $C+1$.

We verify the approximation for the unknown density.  If
$\lambda<4\nu_{\min}$, then
\[
 \lambda^3<512n^{-3a}.
\]
The last assertion of Theorem~\ref{thm:anss_lcs_endpoint} therefore gives
$Z\ge(c_{\mathrm A}/512)\lambda^3\lcs(A,B)$ for all sufficiently large
$n$. Suppose instead that $\lambda\ge4\nu_{\min}$. Use the same pad-first
wrapper as Algorithm~\ref{alg:approximate_lcs}, but substitute the fixed
values of $s,a,x,$ and $\gamma$ from
Lemma~\ref{lem:one_level_adaptive}.  Its maximum window length is $o(n)$,
so the padded length satisfies $n\le n^+<2n$ and padding leaves
$L:=\lcs(A,B)$ unchanged.  Hence
\[
 \lambda^+:=L/n^+\ge\lambda/2\ge2\nu_{\min},
\]
and $\mathcal G$ contains a guess $\nu$ with
$\nu\le\lambda^+<2\nu$.  The successful orientation supplied by
Lemma~\ref{lem:rsss_window_compatible}, followed by
Lemmas~\ref{lem:adaptive_step_one}, \ref{lem:capped_repair}, and
\ref{lem:one_level_adaptive}, returns at least
$c_1\nu^4n^+$ for an absolute $c_1>0$.  Since
$\nu\ge\lambda^+/2\ge\lambda/4$, this is at least
\[
 \frac{c_1}{256}\lambda^4n
 =\frac{c_1}{256}\lambda^3L.
\]

On the intersection of the two branchwise good events, every endpoint and
window value used by the final maximum is at most $L$.  A union bound makes
the failure probability at most $2n^{-(C+1)}\le n^{-C}$.  The ANSS endpoint
takes $O_C(n^{4/3}\log n)$ time and is dominated by the window exponent.
Taking $c:=\min\{c_{\mathrm A}/512,c_1/256\}$, and solving the finitely many
exceptional input lengths exactly, completes the proof.
\end{proof}

The final theorem lets the scale count grow with $n$ and therefore defines one
algorithm rather than a family indexed by an accuracy parameter. It achieves the
same one-sided guarantee in time $n^{8/5+O(1/\sqrt{\log\log n})}$, which is
$n^{8/5+o(1)}$.

\begin{theorem}[Uniform $8/5+o(1)$ bound]
\label{thm:uniform_eight_fifths_formal}
There is an absolute constant $c>0$ such that, for every fixed $C>0$,
Algorithm~\ref{alg:approximate_lcs}, given $A,B\in\Sigma^n$ and
$\lambda=\lcs(A,B)/n$, returns a number $Z$ satisfying
\[
 c\lambda^3\lcs(A,B)\le Z\le\lcs(A,B)
\]
with probability at least $1-n^{-C}$, in time
\[
 n^{8/5+o(1)}.
\] 
\end{theorem}
\begin{proof}
Algorithm~\ref{alg:approximate_lcs} uses the fixed integer $K_\star$ chosen
above from the absolute constant $K_{\mathrm A}$ of
Theorem~\ref{thm:anss_lcs_endpoint}. Put $h:=\log\log n$. Its exact integer
surrogate $H=\lceil\log_2(\lg n)\rceil$ satisfies $H=\Theta(h)$, and for all
sufficiently large $n$ the algorithm chooses
\[
 s:=\lfloor\frac{\sqrt H}{132K_\star}\rfloor,
 \qquad a:=\frac1{66s}.
\]
For the finitely many $n$ for which $s<2$, compute LCS exactly.  The choices
are made by bit lengths and integer square root, so this defines one algorithm
rather than a family indexed by an accuracy parameter. By the choice of
$K_\star$ and the fixed comparison between $H$ and $h$, for all sufficiently
large $n$,
\[
 \frac{2K_{\mathrm A}}{\sqrt h}\le a,
 \qquad a=O(\frac{K_\star}{\sqrt h}),
 \qquad s=\Theta(\sqrt h),
\]
and $a\le1/(33s)$.  Lemma~\ref{lem:one_level_adaptive} gives
\begin{align*}
 F_{s,a}-\frac85
 &\le\frac4{25s}+\frac{21}{5}a
 <15a
 =O(\frac{K_\star}{\sqrt h}).
\end{align*}
Its uniform factor $s^{O(1)}(\log n)^{O(1)}$, together with the density
grid and orientations, is $n^{o(1/\sqrt h)}$.

The algorithm defines $\mathcal G$ and $\nu_{\min}$ exactly as in the proof of
Theorem~\ref{thm:fixed_delta_formal}, now with these $n$-dependent values of
$s$ and $a$. It runs all window guesses and independently calls
Algorithm~\ref{alg:quantitative_small_lambda_estimate}, which implements
Theorem~\ref{thm:anss_lcs_endpoint}, and returns the maximum. If
$\lambda<4\nu_{\min}$, then
$\lambda^3<512n^{-3a}$, while
\[
 \frac{K_{\mathrm A}}{\sqrt h}\le\frac a2<3a.
\]
Consequently the ANSS endpoint returns at least
$(c_{\mathrm A}/512)\lambda^3L$.  If
$\lambda\ge4\nu_{\min}$, the identical padding and dyadic-guess argument in
the preceding proof gives a valid guess and the window branch returns at
least $(c_1/256)\lambda^3L$.  The hypotheses of
Lemma~\ref{lem:one_level_adaptive}, including $s\le\sqrt{\log n}$, hold for
all sufficiently large $n$.

The ANSS endpoint costs only $O_C(n^{4/3}\log n)$.  Including the uniform
lower-order factors and the finitely many rounded lengths, the window
calculation above is
\[
 n^{8/5+O(1/\sqrt h)}
 =n^{8/5+O(1/\sqrt{\log\log n})}
 =n^{8/5+o(1)}.
\]  The same branchwise good events and
union bound as in the fixed-accuracy proof give probability at least
$1-n^{-C}$ and preserve the upper bound.  Taking
$c:=\min\{c_{\mathrm A}/512,c_1/256\}$ proves the theorem.
\end{proof}

\begin{remark}[Status of the exponent calculations]
\label{rem:status_bounds}
Theorem~\ref{thm:explicit_anss_free_formal} uses the fixed three-scale adaptive window
algorithm and the BCAD tradeoff only. Lemma~\ref{lem:main} records the earlier
$116/63$ bound with oblivious Step~1. The cost family in
Lemma~\ref{lem:one_level_adaptive}, combined with the independently amplified
ANSS endpoint of Theorem~\ref{thm:anss_lcs_endpoint}, gives the fixed-accuracy
Theorem~\ref{thm:fixed_delta_formal} and the uniform
Theorem~\ref{thm:uniform_eight_fifths_formal}. The first uses a fixed number
of scales; Algorithm~\ref{alg:approximate_lcs} implements the second using
$\Theta(\sqrt{\log\log n})$ scales. The ANSS
probability guarantee is invoked as the published black box, while its
expected time is converted to a worst-case cap in
Section~\ref{sec:anss_endpoint}.
\end{remark}

\begin{remark}
\label{rem:tiny_delta}
For a fixed scale count $s$, the window cost at a vanishing crossover is
$F_s=(8s+4)/(5s+2)=8/5+4/(25s+10)$. Formally this tends to $8/5$ as
$s\to\infty$, and the limiting costs are $2-x+x\gamma$, $2x$, and
$2-x\gamma$. Their maximum is at least $8/5$, with equality only at
$x=4/5$, $\gamma=1/2$. Theorem~\ref{thm:fixed_delta_formal} first fixes $s$
as a function of $\delta$. Theorem~\ref{thm:uniform_eight_fifths_formal}
instead takes $s=\Theta(\sqrt{\log\log n})$ and balances the moving crossover
with the quantitative ANSS endpoint, turning this limiting calculation into
the stated $8/5+O(1/\sqrt{\log\log n})$ theorem. The fixed three-scale BCAD
balance is retained separately as the explicit ANSS-free
Theorem~\ref{thm:explicit_anss_free_formal}.
\end{remark}

\begin{remark}
\label{rem:status}
Section~\ref{sec:accounting} and Corollary~\ref{cor:r} use the stated results
of \cite{rsss19,bcd21}. Lemmas~\ref{lem:bipartite}, \ref{lem:perpair},
\ref{lem:step2}, \ref{lem:sound}, \ref{lem:complete} and~\ref{lem:round} are
modifications proved here. In particular the exact surrogate of
Lemma~\ref{lem:sound} certifies the RSSS constant $1/16$, whereas the sampled
matrix-product test deliberately writes the weaker one-sided value
$D_{ij}\lambda^4/32$. Dyadic layers are used only as buckets, and every
post-product comparison uses the actual $D_{ij}$. The restriction to
$W_A\times W_B$ is justified by Algorithms~6--7 of \cite{rsss19}, which read
and repair only mixed entries. Lemmas~\ref{lem:staircase_cap}
and~\ref{lem:capped_repair} are new. The first forms one global
heterogeneous-density staircase target and bounds the deficient set; the
second reaches that target by multiscale exact localization: sampled exact
detection identifies deficient cells, each scale shrinks the candidate region
around them, and only the final region is repaired, exhaustively and exactly.
At one scale the construction specializes to the sampled nearby repair of
\cite[Lemma~5.1]{rsss19} with an optimized radius. No approximate recursion is
used at any scale, so the approximation constant is absolute.
\end{remark}

\bibliographystyle{alpha}
\bibliography{ref}

\appendix

\section*{Acknowlegements and AI Disclosure}
The AI tools used in preparing this draft are Gemini Pro 3.1, ChatGPT 5.6, and Codex 5.6. These tools were used mainly to improve the presentation and grammar of the paper.

\end{document}